\documentclass[journal, headsepline]{IEEEtran}

\usepackage{amsmath,amsfonts,amssymb,amsthm}
\usepackage{mathrsfs}
\usepackage{comment,blkarray}
\usepackage{multirow,bigdelim}
\usepackage{cite}
\usepackage{tikz}
\usetikzlibrary{arrows,automata}
\usepackage[latin1]{inputenc}
\usepackage{verbatim}
\usepackage{graphicx}
\usepackage{cases}		
\usepackage{booktabs}
\usepackage{caption}							
\usepackage{subcaption}	
\usepackage{etoolbox}

\usepackage{enumitem}
\usepackage{mathtools}
\usepackage{algorithm}
\usepackage{algpseudocode}
\usepackage{url}

\newtheorem{thm}{Theorem}
\newtheorem{lemma}{Lemma}
\newtheorem{cor}{Corollary}
\newtheorem{prop}{Proposition}
\newtheorem{question}{Question}

\theoremstyle{definition} 
\newtheorem{example}{Example}

\allowdisplaybreaks

\begin{document}

\title{A Generalized Multidimensional Chinese Remainder Theorem (MD-CRT) for Multiple Integer Vectors}

\author{Guangpu Guo and Xiang-Gen Xia,
\IEEEmembership{Fellow}, \IEEEmembership{IEEE} 
\thanks{G. Guo and X.-G. Xia are with the Department of Electrical and Computer Engineering,
  University of Delaware, Newark, DE 19716, USA
  (e-mails: guangpu@udel.edu and xxia@ece.udel.edu).
  This work was supported in part by the National Science Foundation (NSF) under Grant CCF-2246917.
}
}
\maketitle

\begin{abstract}
  Chinese remainder theorem (CRT) is widely applied in cryptography, coding theory, and signal processing. It has been extended to the multidimensional CRT (MD-CRT), which reconstructs an integer vector from its vector remainders modulo multiple integer matrices. This paper investigates a generalized MD-CRT for multiple integer vectors, where the goal is to determine multiple integer vectors from multiple vector residue sets modulo multiple integer matrices.
  Comparing to the existing generalized CRT for multiple scalar integers, the challenge is that the moduli in MD-CRT are matrices that do not commute and the corresponding uniquely determinable range is multidimensional and the inclusion relationship is much more complicated. In this paper,
  we address two fundamental questions regarding the generalized MD-CRT. The first question concerns the uniquely determinable range of multiple integer vectors when no prior information about them is available. The second question is about the conditions under which the maximal possible dynamic range can be achieved. 
  To answer these two questions, we first derive a uniquely determinable range without prior information and accordingly propose an algorithm  to achieve it. A special case involving only two integer vectors is investigated for the second question, leading to a new condition for achieving the maximal possible dynamic range.
  Interestingly, this newly obtained condition, when the dimension is reduced to $1$, is even better
  than the existing ones for the conventional generalized CRT for scalar integers. 
  These results may have applications for frequency detection in multidimensional signal processing.
\end{abstract}

\begin{IEEEkeywords}
Chinese remainder theorem (CRT), multidimensional Chinese remainder theorem (MD-CRT), multiple frequency detection, dynamic range, undersampling.
\end{IEEEkeywords}

\IEEEpeerreviewmaketitle

\section{Introduction}\label{s1}

Chinese remainder theorem (CRT) \cite{crt1} has many applications in, for example, cryptography,  coding theory, and signal processing \cite{crt1, crt2, radar_book,xia99,gangli1,fft2,wenchaoli,radeee,xiaoli,congling1,congling2,congling3,lugan2}.
It is to reconstruct a large positive integer $N$ from its remainders modulo several positive integer moduli $0<m_1<m_2<\cdots<m_{\gamma}$ and the large integer $N$ can be uniquely reconstructed if and only if it is smaller than the least common multiple (lcm), called the dynamic range, of all the moduli. 
Many different extensions have been proposed in the past, including CRT for multiple integers \cite{xia99,zhou,xia00,liao07,xiao14,wang15,xiao15}, called generalized CRT, and robust CRT \cite{xia07,li09,wang10,xiao2014,xiao17,li16,xhs18}. 
CRT has been also  extended in \cite{ar,PPV1,MD1,MD2} from scalar integers to vectors of integers, called multidimensional CRT (MD-CRT), which is to reconstruct an integer vector from its vector remainders modulo multiple integer matrices, referred to as matrix moduli. An integer vector can be uniquely reconstructed if and only if it is in the fundamental parallelepiped (FPD) of a least common right multiple (lcrm) of all the matrix moduli \cite{MD1, MD2}. This FPD of an lcrm is a uniquely determinable range and the size of this FPD, i.e., the number of integer vectors in FPD, is called the dynamic range as well, which is equal to the absolute determinant value of the lcrm of all the matrix moduli. The MD-CRT in \cite{MD1,MD2} is for general integer matrix moduli, while the MD-CRT in \cite{ar,PPV1} is for integer matrix moduli that can be simultaneously diagonalized. A novel family of pairwise co-prime integer matrices of any dimension was recently obtained in \cite{guo25}.

In this paper, motivated by the multiple frequency detection of multidimensional undersampled waveforms that may occur in applications as we shall see in Section \ref{s2}, we investigate a generalized MD-CRT for multiple integer vectors, which has not been addressed before. Unlike the MD-CRT for a single integer vector, this problem involves reconstructing multiple integer vectors from several sets of vector remainders called vector residue sets. The vector remainders in each vector residue set are derived from all the integer vectors, needed to determine, modulo the same matrix modulus. However, the correspondence between the vector remainders within the same vector residue set and the vectors needed to determine is unknown. To address the ambiguity caused by the unknown correspondence between the vectors and their vector remainders in each vector residue set, the uniquely determinable range for multiple vectors may be smaller than that of MD-CRT for a single integer vector, if there is no prior information to help clarify the unknown correspondence, thus leads to a lower dynamic range. In addition, it may be possible that the same dynamic range as MD-CRT for a single integer vector, i.e., the maximal possible dynamic range of MD-CRT for multiple integer vectors, can be reached, if we know some prior information on the vectors that we want to determine.  

Two questions about the generalized MD-CRT are raised. The first question is what is a uniquely determinable range of MD-CRT for multiple integer vectors if there is no prior information on these vectors? The second question is under what conditions the maximal possible dynamic range of MD-CRT for multiple integer vectors can be reached?

In this paper, we first present a uniquely determinable range of MD-CRT for multiple integer vectors, which gives an answer to the first question. An efficient algorithm is also introduced correspondingly. These results, when reduced to one dimension, are consistent with the result of generalized CRT for multiple integers in \cite{xia99,xia00}. It is known from the results in \cite{xia99,xia00} that for the one dimensional generalized CRT, the more unknown integers to determine, the less dynamic range is.
In other words, a uniquely determinable region for more unknown integers to determine is included in that for less unknown integers to determine.
This may not be true in general for generalized MD-CRT studied in this paper, which is because the FPD is multidimensional and two different FPDs may not have a strict
inclusion relationship as the one dimensional case. It is a key difference with the existing generalized CRT for multiple scalar integers, in, for example, \cite{xia99,zhou,xia00,liao07,xiao14,wang15,xiao15}.

As for the second question, we consider a special case of MD-CRT for two integer vectors and present a condition to achieve the maximal possible dynamic range with a detailed determination algorithm.
Interestingly, this condition, when reduced to one dimension, is even better than all the previous results in \cite{zhou,xiao15}. The previous results for one dimensional generalized CRT indicate that the maximal possible dynamic range can be achieved when the difference between the two unknown integers is small. The new result obtained in this paper, however, demonstrates that the maximal possible dynamic range can be achieved not only when the difference is small but also when the difference is large and close to the dynamic range. Moreover, the condition on the difference of two integers we obtain is strictly weaker than that obtained in \cite{zhou,xiao15}. In addition, when a size constraint is imposed on the moduli, that is, the moduli must be smaller than a certain number, the condition proposed in this paper is not only strictly weaker than those in \cite{xiao15} but also yields a larger dynamic range.

The remaining of this paper is organized as follows. In Section \ref{s2}, we first introduce some necessary preliminaries. We also present the motivated problem of multiple frequency detection of undersampled multidimensional waveforms and formalize the questions that we address in this paper. In Section \ref{s3}, we focus on the generalized MD-CRT for multiple integer vectors without prior information and present a uniquely determinable range for multiple integer vectors. In Section \ref{s4}, we consider the generalized MD-CRT for two integer vectors with prior information and a condition to achieve the maximal possible dynamic range is given. We also analyze the advantages of our result, when reduced to one dimension, over all the previous results. In Section \ref{s5}, we conclude this paper.

\section{Preliminaries and Problem Formulation}\label{s2}

We first have some notations. $\mathbb{Z}$ denotes the set of all integers and $\mathbb{R}$ denotes the set of all real numbers. Lowercase boldface letters represent vectors, while uppercase boldface letters represent matrices. All vectors and matrices in this paper are $D$ dimensional integer vectors and $D\times D$ dimensional integer matrices, respectively, unless otherwise specified. And diag stands for a diagonal matrix, $^{\top}$ means the transpose of vectors or matrices. Below we introduce some necessary concepts on integer matrices and for details, see, for example,
\cite{PPV1, MD1, MD2, matrix, remainder}.

\subsection{Preliminaries}

1) \textbf{Least common right multiple (lcrm)}: A non-singular integer matrix $\mathbf{A}$ is a right multiple of an integer matrix $\mathbf{M}$, if there exists a non-singular integer matrix $\mathbf{P}$ such that $\mathbf{A} = \mathbf{M}\mathbf{P}$. If $\mathbf{A}$ is a right multiple of each of all $\gamma \geq 2$ integer matrices $\mathbf{M}_1, \mathbf{M}_2, \dots, \mathbf{M}_{\gamma}$, $\mathbf{A}$ is called a common right multiple (crm) of $\mathbf{M}_1, \mathbf{M}_2, \dots, \mathbf{M}_{\gamma}$. Additionally, $\mathbf{A}$ is a {\em least common right multiple} (lcrm) of $\mathbf{M}_1, \mathbf{M}_2, \dots, \mathbf{M}_{\gamma}$, if any other crm of them is a right multiple of $\mathbf{A}$. If $\mathbf{A}$ is an lcrm of a group of integer matrices, $\mathbf{AU}$ is also an lcrm of them when $\mathbf{U}$ is a unimodular matrix, i.e., $\mathbf{U}$ is an integer matrix and its determinant is either $1$ or $-1$, which means that lcrm is not unique but the absolute determinant value of lcrm is unique. A detailed algorithm to calculate an lcrm of two integer matrices can be found in \cite{MD2,remainder}.

2) \textbf{Lattice}: For a non-singular integer matrix $\mathbf{M}$, its associated lattice is defined as follows:
\[
LAT(\mathbf{M})=\{\mathbf{M}\mathbf{n} \mid \mathbf{n}\in \mathbb{Z}^{D}\},
\]
and a scaled lattice is defined as
\[
\lambda LAT(\mathbf{M})=\{\lambda\mathbf{k} \mid \mathbf{k}\in LAT(\mathbf{M}) \},
\]
where $\lambda$ is a positive real number. For a positive integer $a$, when $\lambda=1/a$, $LAT(\mathbf{M})\subset \lambda LAT(\mathbf{M})$; when $\lambda=a$, $\lambda LAT(\mathbf{M})\subset LAT(\mathbf{M})$.

\begin{prop}\cite{matrix}\label{samelattice}
    Given two non-singular integer matrices $\mathbf{M}_1$ and $\mathbf{M}_2$, $LAT(\mathbf{M}_1) = LAT(\mathbf{M}_2)$ if and only if $\mathbf{M}_1^{-1}\mathbf{M}_2$ is unimodular, i.e., it is an integer matrix with determinant $1$ or $-1$.
\end{prop}

\begin{prop}\cite{matrix}\label{caplattice}
   Let $\mathbf{M}_j$ for $j=1,2,\cdots,\gamma$, be $\gamma$ non-singular integer matrices and $\mathbf{R}$ is an lcrm of all the matrices $\mathbf{M}_j$. Then, $LAT(\mathbf{R})=\bigcap\limits_{j=1}^{\gamma}LAT(\mathbf{M}_j)$.
\end{prop}

3) \textbf{Fundamental parallelepiped (FPD)}: For a non-singular integer matrix $\mathbf{M}$, the fundamental parallelepiped of $\mathbf{M}$ is defined as follows:
\[
\mathcal{N}(\mathbf{M})= \left\{ \mathbf{k} \mid \mathbf{k} = \mathbf{M} \mathbf{x}, \mathbf{x} \in [0,1)^D \text{and} \ \mathbf{k} \in \mathbb{Z}^D \right\}.
\]
There are totally $|\det(\mathbf{M})|$ vectors in $\mathcal{N}(\mathbf{M})$. In one dimensional case, for a positive integer $m$, $\mathcal{N}(m)=\{0,1,2,\cdots,m-1\}$. 

4) \textbf{Division representation of integer vectors}: For any given integer vector $\mathbf{f}$, it has a unique representation with respect to a non-singular integer matrix $\mathbf{M}$:
\[
\mathbf{f}= \mathbf{M}\mathbf{n}+\mathbf{r}, \quad \mathbf{r} \in \mathcal{N}(\mathbf{M}),
\]
where $\mathbf{n}$ and $\mathbf{r}$ are integer vectors and $\mathbf{r}$ is called the vector remainder of $\mathbf{f}$ modulo $\mathbf{M}$. In applications, we take much more care of $\mathbf{r}$ and a formulation that is much easier to implementation to calculate $\mathbf{r}$ is given in \cite{remainder}:
\begin{equation}\label{remaider}
    \mathbf{r} = \mathbf{M} \left( \text{adj}(\mathbf{M}) \mathbf{m} \mod \det(\mathbf{M}) \right) / \det(\mathbf{M}),
\end{equation}
where the modulo operation is performed element-wisely. 

\subsection{Problem Formulation}

Consider the following multiple frequency multidimensional signal:
\begin{equation}\label{signal}
    x(\mathbf{t}) = \sum_{i=1}^{\rho} a_i \exp{(j2\pi \mathbf{f}_{i}^{\top}\mathbf{t})} + \omega(\mathbf{t}), \quad \mathbf{t} \in \mathbb{R}^{D},
\end{equation}
where $a_i$ is an unknown amplitude and $\mathbf{f}_i$ is a $D$ dimensional frequency vector, which is assumed to be an integer vector, for each $i = 1, 2,\cdots,\rho$. Additionally, $\omega(\mathbf{t})$ is an additive noise. We want to determine all the $\rho$ frequency vectors, $\mathbf{f}_i$, from multiple undersampled $D$ dimensional signals of $x(\mathbf{t})$ with sampling rates as low as possible as follows.

We first use $\gamma$ many $D\times D$ non-singular integer matrices $\mathbf{M}_1,\cdots,\mathbf{M}_{\gamma}$, called {\em sampling matrices}, to sample the signal in (\ref{signal}) and get the following $\gamma$ sampled signals:
\begin{equation}\label{md-undersample}
x_{j}[\mathbf{n}] = \sum_{i=1}^{\rho} a_i \exp{(j2\pi \mathbf{f}_{i}^\top \mathbf{M}_j^{-\top} \mathbf{n})} + \omega[\mathbf{M}_j^{-\top} \mathbf{n}], 
\end{equation}
$\mathbf{n}\in \mathbb{Z}^{D}, \ j=1,2,\cdots,\gamma.$ For each sampling matrix $\mathbf{M}_j$,
there are total $|\det (\mathbf{M}_j)|$ many sampled points per unit spatial volume of $\mathbb{R}^{D}$, which is called the {\em sampling rate} (or sampling density) of sampling matrix $\mathbf{M}_j$ for a multidimensional signal. 

Next, by applying the multidimensional DFT (MD-DFT) to each sampled signal $x_j[\mathbf{n}]$ with respect to $\mathbf{n} \in \mathcal{N}(\mathbf{M}_j^{\top})$, we obtain, for $\mathbf{k}\in \mathcal{N}(\mathbf{M}_j)$,
\begin{equation}
\begin{aligned}
    &X_j(\mathbf{k}) \\ 
    = &\sum_{i=1}^{\rho} \sum_{\mathbf{n} \in \mathcal{N}(\mathbf{M}_j^{\top})} a_i \exp(j 2 \pi \mathbf{f}_i^{\top} \mathbf{M}_j^{-\top} \mathbf{n}) \exp(-j 2 \pi \mathbf{k}^{\top} \mathbf{M}_j^{-\top} \mathbf{n})\\
   &+ \Omega_i(\mathbf{k}), \quad j=1,2,\cdots,\gamma,
\end{aligned}
\end{equation}
which can be simplified as \cite{ar} 
\begin{equation}\label{mddft}
    X_j(\mathbf{k}) = \sum_{i=1}^{\rho} a_i |\det(\mathbf{M}_j)| \delta(\mathbf{k} - \mathbf{r}_{i,j}) + \Omega_i(\mathbf{k}), 
\end{equation}
for $j=1,2,\cdots,\gamma$, where $\mathbf{r}_{i,j}$ is the integer vector remainder of
the integer frequency vector $\mathbf{f}_i$ modulo $\mathbf{M}_j$,
and $\delta (\mathbf{n})$ is the discrete delta function that is $1$ when
$\mathbf{n}={\bf 0}$ and $0$ otherwise.
From (\ref{mddft}), one can detect $\rho$ integer vector remainders
$\mathbf{r}_{i,j} \equiv \mathbf{f}_i \mod \mathbf{M}_j$, 
for $i=1,2,\cdots,\rho$, from each sampled signal $X_j(\mathbf{k})$. The above undersampling may occur in practical applications. One example is in SAR imaging of fast moving targets using planar antenna arrays. It is similar to the SAR imaging of fast moving targets in SAR imaging using linear co-prime arrays \cite{gangli1}, where the number of antennas determines the parameter estimation accuracy. Unfortunately  linear antenna arrays are not spatially efficient and planar antenna arrays may pack more antennas in a fixed radar platform space, such as an aircraft, and thus may improve the parameter estimation  accuracy. When co-prime multidimensional antenna arrays are used, the signal model (\ref{md-undersample}) may occur for the radar return signals after some radar signal processing, such as range compression. Another example is the recently active self-reset analog-to-digital converter (SR-ADC) using multi-channels for complex valued bandlimited signals \cite{lugan2}, which can be thought of as a special case of 2 dimensional undersampling.

Let
\begin{equation}
   \mathcal{S}_j(\mathbf{f}_1,\cdots,\mathbf{f}_{\rho})=\{\mathbf{r}_{i,j}\ |\ i=1,2,\cdots,\rho\} 
\end{equation}
be the set of all the detected integer vector remainders from the above $X_{j}(\mathbf{k})$, called vector residue set, for $j=1,2,\cdots,\gamma$, and
\begin{equation}
   \mathcal{S}(\mathbf{f}_1,\cdots,\mathbf{f}_{\rho})=\mathcal{S}_1(\mathbf{f}_1,\cdots,\mathbf{f}_{\rho}) \times \cdots \times \mathcal{S}_{\gamma}(\mathbf{f}_1,\cdots,\mathbf{f}_{\rho}), 
\end{equation}
which is called the residue set of $\mathbf{f}_i$, $1 \leq i \leq \rho$, modulo $\mathbf{M}_j$, $1 \leq j \leq \gamma$.
Now the problem becomes how to determine $\rho$ integer vectors $\mathbf{f}_1,\cdots,\mathbf{f}_{\rho}$ from $\mathcal{S}(\mathbf{f}_1,\cdots,\mathbf{f}_{\rho})$.

When $\rho=1$, there is only one frequency vector $\mathbf{f}$ in (\ref{signal}) to determine, which becomes the MD-CRT problem as follows.   

\begin{prop}\label{MD-CRT}(MD-CRT\cite{MD1})
    Given $\gamma$ matrix moduli $\mathbf{M}_j$ for $1 \leq j \leq \gamma$, which are arbitrary non-singular integer matrices. Let $\mathbf{R}$ be an lcrm of these $\gamma$ matrix moduli. For an integer vector $\mathbf{f} \in \mathbb{Z}^D$,  
it can be uniquely determined from its $\gamma$ integer vector remainders $\mathbf{r}_j \equiv \mathbf{f} \mod \mathbf{M}_j$, $1\leq j\leq \gamma$, if
and only if 
$\mathbf{f} \in \mathcal{N}(\mathbf{R})$.
\end{prop} 

The sufficiency part is obtained in \cite{MD1}. The necessity part is based on the fact that if  
$\mathbf{f} \notin \mathcal{N}(\mathbf{R})$, by following the MD-CRT algorithm in \cite{MD1} we can obtain a vector $\mathbf{f}^{\prime} \in \mathcal{N}(\mathbf{R})$ such that $\mathbf{f}$ and $\mathbf{f}^{\prime}$ have the same vector remainders $\mathbf{r}_1,\cdots,\mathbf{r}_{\gamma}$ modulo matrix moduli $\mathbf{M}_1,\cdots,\mathbf{M}_{\gamma}$, respectively. 

For the reconstruction uniqueness, similar to the classical CRT, an lcrm associated with MD-CRT should be selected in advance. In the remaining paper, all involved lcrms are assumed to be fixed in advance.

From Prop. \ref{MD-CRT}, $\mathcal{N}(\mathbf{R})$ is a uniquely determinable range (or region) of $\mathbf{f}$. Although $\mathcal{N}(\mathbf{R})$ may not be the same for different lcrms of all the matrix moduli, 
the numbers of vectors in different sets $\mathcal{N}(\mathbf{R})$, i.e., the sizes of uniquely determinable ranges for the frequency vector, are always the same and equal to $|\det(\mathbf{R})|$, which is called the \textit{dynamic range} of the MD-CRT.

When $\rho >1$, the problem of determining $\mathbf{f}_1,\cdots,\mathbf{f}_{\rho}$ from $\mathcal{S}(\mathbf{f}_1,\cdots,\mathbf{f}_{\rho})$ is referred to as {\em generalized MD-CRT for multiple integer vectors}. In this case, there may be more than one vector in $\mathcal{S}_j(\mathbf{f}_1,\cdots,\mathbf{f}_{\rho})$.
Since, for each $j$, we detect all integer vector remainders in $\mathcal{S}_j(\mathbf{f}_1,\cdots,\mathbf{f}_{\rho})$ simultaneously from $X_{j}(\mathbf{k})$ using MD-DFT in (\ref{mddft}), we don't know which vector in this set corresponds to which $\mathbf{f}_i$ as its vector remainder, i.e., the correspondence between the vector remainders in $\mathcal{S}_j(\mathbf{f}_1,\cdots,\mathbf{f}_{\rho})$ and the vectors $\mathbf{f}_i$, $i=1,2,\cdots,\rho$, that we want to determine, is unknown. The difficulty of generalized MD-CRT for multiple integer vectors is how to address the unknown correspondence.

If there is no prior information on $\{\mathbf{f}_1,\cdots,\mathbf{f}_{\rho}\}$, we can't clarify the unknown correspondence in general. Due to the ambiguity caused by the unknown correspondence in each vector residue set $\mathcal{S}_j(\mathbf{f}_1,\cdots,\mathbf{f}_{\rho})$, the uniquely determinable range of $\{\mathbf{f}_1,\cdots,\mathbf{f}_{\rho}\}$ in this problem may be smaller than that of MD-CRT for a single vector, which also means a lower dynamic range. We first have the following question.

\begin{question}\label{q1}
    What is a uniquely determinable range of $\{\mathbf{f}_1,\cdots,\mathbf{f}_{\rho}\}$, if there is no prior information on $\{\mathbf{f}_1,\cdots,\mathbf{f}_{\rho}\}$?  
\end{question}

On the other hand, it may be possible that we can specify the correspondence in each vector residue set $\mathcal{S}_j$, if we know some prior information on $\{\mathbf{f}_1,\cdots,\mathbf{f}_{\rho}\}$. In this case, the problem becomes how to determine each $\mathbf{f}_i$ from its vector remainders $\mathbf{r}_{i,j}$ and can be solved by using MD-CRT separately for each $\mathbf{f}_i$. Then, we can uniquely determine all of them if and only if $\{\mathbf{f}_1,\cdots,\mathbf{f}_{\rho}\} \subset \mathcal{N}(\mathbf{R})$. In this case, the maximal possible dynamic range is achieved. Thus, we have the following question.

\begin{question}\label{q2}
    Under what prior information on $\{\mathbf{f}_1,\cdots,\mathbf{f}_{\rho}\}$, the maximal possible dynamic range can be achieved?
\end{question}

In the following sections, we investigate these two questions and present answers to them.

\section{MD-CRT for Multiple Integer Vectors without Prior Information}\label{s3}

In this section, we focus on Question \ref{q1} and present a uniquely determinable range of MD-CRT for multiple integer vectors. An efficient algorithm is also introduced correspondingly with two detailed examples. 

For simplicity, let $\mathbf{R}_{\mathcal{A}}$ be an lcrm of $\{\mathbf{M}_{j}:j\in \mathcal{A}\}$ for any $\mathcal{A} \subset \{1,2,\cdots,\gamma\}$ and 
\begin{equation}\label{Neta}
    \mathcal{N}_{\eta}=\bigcap\limits_{\mathcal{A} \subset \{1,2,\cdots,\gamma\}, |\mathcal{A}|=\eta} \mathcal{N}(\mathbf{R}_{\mathcal{A}})
\end{equation}
for $1\leq \eta \leq \gamma$, we then have the following result.

\begin{thm}\label{th1}
    For any given $\gamma$ integer matrix moduli $\mathbf{M}_j$, $j=1,2,\cdots,\gamma$, $\rho$ integer vectors $\mathbf{f}_i$, $i=1,2,\cdots,\rho$, can be uniquely determined from their residue set $\mathcal{S}(\mathbf{f}_1,\mathbf{f}_2,\cdots,\mathbf{f}_{\rho})$, if 
    \begin{itemize}
        \item[1)] $\{\mathbf{f}_1,\mathbf{f}_2,\cdots,\mathbf{f}_{\rho}\} \subset \mathcal{N}(\mathbf{M}_j)$\ \text{for some} $j=1,2,\cdots,\gamma$, \quad \text{or} 
        \item[2)] $\{\mathbf{f}_1,\mathbf{f}_2,\cdots,\mathbf{f}_{\rho}\} \subset \mathcal{N}_{\eta}$, \quad $\eta \geq 1$,
    \end{itemize}
    where $\eta$ satisfies $\gamma=\eta\rho + \alpha$ and $0\leq \alpha < \rho$.
\end{thm}

\begin{proof}
    If there exists some $j=1,2,\cdots,\gamma$, such that $\{\mathbf{f}_1,\mathbf{f}_2,\cdots,\mathbf{f}_{\rho}\} \subset \mathcal{N}(\mathbf{M}_j)$, then we have $\mathbf{r}_{i,j}=\mathbf{f}_i$ for all $i=1,2,\cdots,\rho$. So, all $\rho$ vectors in the set $\mathcal{S}_j(\mathbf{f}_1,\mathbf{f}_2,\cdots,\mathbf{f}_{\rho})$ are that we want to determine. 

    Next, we consider the second uniquely determinable range. When $\eta=1$, it becomes
    $$\{\mathbf{f}_1,\mathbf{f}_2,\cdots,\mathbf{f}_{\rho}\} \subset \bigcap\limits_{j} \mathcal{N}(\mathbf{M}_{j}) \subset \mathcal{N}(\mathbf{M}_j),$$
    which is included in the first uniquely determinable range. 

    When $\eta>1$, we prove that $\mathcal{N}_{\eta}$ is a uniquely determinable range for $\mathbf{f}_1,\mathbf{f}_2,\cdots,\mathbf{f}_{\rho}$ by proving that if there are two sets of integer vectors $\{\mathbf{f}_1,\cdots,\mathbf{f}_{\rho}\}$ and $\{\mathbf{k}_1,\cdots,\mathbf{k}_{\rho}\}$ in $\mathcal{N}_{\eta}$ such that $$\mathcal{S}(\mathbf{f}_1,\cdots,\mathbf{f}_{\rho}) = \mathcal{S}(\mathbf{k}_1,\cdots,\mathbf{k}_{\rho}),$$
    then we have $\{\mathbf{f}_1,\cdots,\mathbf{f}_{\rho}\}=\{\mathbf{k}_1,\cdots,\mathbf{k}_{\rho}\}$.

    Assume that $\{\mathbf{f}_1,\cdots,\mathbf{f}_{\rho}\}$ and $\{\mathbf{k}_1,\cdots,\mathbf{k}_{\rho}\}$ are two sets of vectors in $\mathcal{N}_{\eta}$ such that $$\mathcal{S}(\mathbf{f}_1,\cdots,\mathbf{f}_{\rho}) = \mathcal{S}(\mathbf{k}_1,\cdots,\mathbf{k}_{\rho}).$$
    For each given $\mathbf{f}_i$, $1 \leq i \leq \rho$, there exists a vector $\mathbf{k}_{x_j}$ for $1 \leq x_j \leq \rho$ such that $$\mathbf{f}_i\equiv\mathbf{k}_{x_j} \mod \mathbf{M}_j$$ for every $1\leq j \leq \gamma$. Since $\gamma=\eta\rho + \alpha$ and $\eta>1$, i.e., the number of matrix moduli is $\eta$ times greater than that of vectors that need to be determined, there must exist $\eta$ vectors $\mathbf{k}_{x_{j_1}},\cdots,\mathbf{k}_{x_{j_{\eta}}}$ such that $\mathbf{k}_{x_{j_1}}=\cdots=\mathbf{k}_{x_{j_{\eta}}}$, for $1 \leq j_1,\cdots,j_{\eta} \leq \gamma$ and $1 \leq x_{j_1},\cdots,x_{j_{\eta}} \leq \rho$. Without loss of generality, we denote this vector by $\mathbf{k}_{i^{\prime}}$. It also means that $$\mathbf{f}_i\equiv\mathbf{k}_{i^{\prime}} \mod \mathbf{M}_{j_y}$$ for all $y=1,2,\cdots,\eta$.
    Let $\mathcal{A}=\{j_1,\cdots,j_{\eta}\}$ and $\mathbf{R}_{\mathcal{A}}$ be an lcrm of $\mathbf{M}_{j_y}$ for all $y=1,2,\cdots,\eta$. We then have $$\mathbf{f}_i\equiv\mathbf{k}_{i^{\prime}} \mod \mathbf{R}_{\mathcal{A}}.$$ From the condition 2) above, we have $\mathbf{f}_{i},\mathbf{k}_{i^{\prime}} \in \mathcal{N}(\mathbf{R}_{\mathcal{A}})$. As a result, $\mathbf{f}_i=\mathbf{k}_{i^{\prime}}$. It means that for each given $\mathbf{f}_i$ for $1 \leq i \leq \rho$, we have $\mathbf{f}_i \in \{\mathbf{k}_1,\cdots,\mathbf{k}_{\rho}\}$. So, we can get $\{\mathbf{f}_1,\cdots,\mathbf{f}_{\rho}\} \subset \{\mathbf{k}_1,\cdots,\mathbf{k}_{\rho}\}$. Similarly, we can also get $\{\mathbf{k}_1,\cdots,\mathbf{k}_{\rho}\} \subset \{\mathbf{f}_1,\cdots,\mathbf{f}_{\rho}\}$, which means that $\{\mathbf{f}_1,\cdots,\mathbf{f}_{\rho}\} = \{\mathbf{k}_1,\cdots,\mathbf{k}_{\rho}\}$.
\end{proof}

Since there are many choices of $\mathbf{R}_{\mathcal{A}}$ in (\ref{Neta}) for any subset ${\cal A}$ of $\{1,2,\cdots, \gamma\}$, there are many possible $\mathcal{N}_{\eta}$ as well. Theorem \ref{th1} holds for any possible $\mathcal{N}_{\eta}$.
For any possible $\mathcal{N}_{\eta}$, we have $|\mathcal{N}_{\eta}|\leq |\mathcal{N}(\mathbf{R}_{\mathcal{A}})| \leq |\mathcal{N}(\mathbf{R}_{\{1,2,\cdots,\gamma\}})|$, since $\mathbf{R}_{\{1,2,\cdots,\gamma\}}$ is a crm of $\mathbf{R}_{\mathcal{A}}$ and $\mathbf{M}_j$ for $j \in \{1,2,\cdots,\gamma\} \setminus \mathcal{A}$, where $|\mathcal{N}(\mathbf{R}_{\{1,2,\cdots,\gamma\}})|$ is the dynamic range of MD-CRT for a single integer vector. As a result, the dynamic range obtained in Theorem \ref{th1} is, not surprisingly, less than or equal to that of MD-CRT for a single integer vector, and in most cases, is strictly less than that of MD-CRT, which depends on the choice of matrix moduli and the number $\rho$ of integer vectors need to determine.

Similar to MD-CRT, in this problem, $\mathbf{R}_{\mathcal{A}}$ and consequently $\mathcal{N}_{\eta}$ must be selected in advance. 
In the following discussion, we assume this selection has been made.

When $\{\mathbf{f}_1,\mathbf{f}_2,\cdots,\mathbf{f}_{\rho}\} \subset \mathcal{N}(\mathbf{M}_j) \ \text{for some} \ j=1,2,\cdots,\gamma$, we can get these $\rho$ vectors directly from $\mathcal{S}_j(\mathbf{f}_1,\mathbf{f}_2,\cdots,\mathbf{f}_{\rho})$.
When $\{\mathbf{f}_1,\mathbf{f}_2,\cdots,\mathbf{f}_{\rho}\} \subset \mathcal{N}_{\eta}$, for $1 < \eta \leq \gamma$, 
according to Theorem \ref{th1}, if we calculate the residue set $\mathcal{S}(\mathbf{f}_1,\cdots,\mathbf{f}_{\rho})$ for all possible $\{\mathbf{f}_1,\cdots,\mathbf{f}_{\rho}\}$ to form a table beforehand, then we can search the table to determine $\{\mathbf{f}_1,\cdots,\mathbf{f}_{\rho}\}$ after we detect the residue set $\mathcal{S}(\mathbf{f}_1,\mathbf{f}_2,\cdots,\mathbf{f}_{\rho})$. However, when the size of $\mathcal{N}_{\eta}$ is larger, it is not efficient to calculate all possible residue sets. We next present an efficient algorithm to determine $\{\mathbf{f}_1,\mathbf{f}_2,\cdots,\mathbf{f}_{\rho}\}$ from the set $\mathcal{S}(\mathbf{f}_1,\mathbf{f}_2,\cdots,\mathbf{f}_{\rho})$. To do so, we first present a lemma.

\begin{lemma}\label{lm:anotherf}
    If $\{\mathbf{f}_1,\cdots,\mathbf{f}_{\rho}\} \subset \mathcal{N}_{\eta}$, we can not find another integer vector $\mathbf{f} \in \mathcal{N}_{\eta} \setminus \{\mathbf{f}_1,\cdots,\mathbf{f}_{\rho}\}$ such that $(\mathbf{r}_{1},\cdots,\mathbf{r}_{\gamma}) \in \mathcal{S}(\mathbf{f}_1,\cdots,\mathbf{f}_{\rho})$, where $\mathbf{r}_j$ is the vector remainder of $\mathbf{f}$ modulo $\mathbf{M}_j$ for $1 \leq j\leq \gamma$.
\end{lemma}

\begin{proof}
    Assume $\mathbf{f} \in \mathcal{N}_{\eta}$ and its vector remainders $(\mathbf{r}_1,\cdots,\mathbf{r}_{\gamma}) \in \mathcal{S}(\mathbf{f}_1,\cdots,\mathbf{f}_{\rho})$. If there exists some $i$ for $1\leq i \leq \rho$ such that $$(\mathbf{r}_1,\cdots,\mathbf{r}_{\gamma})=(\mathbf{r}_{i,1},\cdots,\mathbf{r}_{i,\gamma}),$$
    we then have $\mathbf{f}=\mathbf{f}_i$. Otherwise, there must exist $\eta$ vectors $\mathbf{r}_{i_1},\cdots,\mathbf{r}_{i_{\eta}}$, for $1 \leq i_1,\cdots,i_{\eta} \leq \gamma$, such that these $\eta$ vectors are vector remainders of the same vector $\mathbf{f}_i$ modulo $\mathbf{M}_{i_1},\cdots,\mathbf{M}_{i_{\eta}}$ for some $i=1,2,\cdots,\rho$. Since $\mathbf{f}_i \in \mathcal{N}_{\eta} \subset \mathcal{N}(\mathbf{R}_{\{i_1,\cdots,i_{\eta}\}})$, $\mathbf{f}_i$ is the only vector in $\mathcal{N}_{\eta}$ such that $\mathbf{r}_{i_1},\cdots,\mathbf{r}_{i_{\eta}}$ are its vector reminders modulo $\mathbf{M}_{i_1},\cdots,\mathbf{M}_{i_{\eta}}$ according to the MD-CRT in Prop. \ref{MD-CRT}. Since all $\mathbf{r}_j$, for $1 \leq j \leq \gamma$, are vector remainders of the same vector in $\mathcal{N}_{\eta}$, we have $(\mathbf{r}_1,\cdots,\mathbf{r}_{\gamma})=(\mathbf{r}_{i,1},\cdots,\mathbf{r}_{i,\gamma})$, i.e., $\mathbf{f}=\mathbf{f}_i$. Then, we complete this proof.
\end{proof}

Note that Theorem \ref{th1} and Lemma \ref{lm:anotherf} are similar to that in the one dimensional case \cite{xia99} and \cite{xia00}, respectively. Next, we present a detailed algorithm to reconstruct $\{\mathbf{f}_1,\cdots,\mathbf{f}_{\rho}\}$ according to $\mathcal{S}(\mathbf{f}_1,\cdots,\mathbf{f}_{\rho})$ when they are all in $\mathcal{N}_{\eta}$. The above lemma will
be used in the following algorithm development.

\begin{description}[leftmargin=1.26cm]
  \item[Step 1:]
  Arbitrarily choose $(\mathbf{r}_{{i_1},{1}},\cdots,\mathbf{r}_{{i_{\gamma}},{\gamma}})$ from $\mathcal{S}(\mathbf{f}_1,\cdots,\mathbf{f}_{\rho})$.
\end{description}

Since $\gamma=\eta\rho + \alpha$, i.e., the number of moduli is $\eta$ times greater than the number of vectors needed to be reconstructed, there must exist $1\leq j_1,j_2,\cdots,j_{\eta} \leq \gamma$ such that $\mathbf{r}_{i_{j_1},j_1}, \mathbf{r}_{i_{j_2},j_2}, \cdots, \mathbf{r}_{i_{j_{\eta}},j_{\eta}}$ chosen in Step 1 are vector remainders of a single vector, i.e., one of $\mathbf{f}_1,\cdots,\mathbf{f}_{\rho}$, modulo different matrix moduli $\mathbf{M}_{j_1},\cdots,\mathbf{M}_{j_{\eta}}$. Then, this vector can be obtained by using MD-CRT to solve the following congruence system:
\begin{equation}\label{eta}
    \begin{aligned}
      \mathbf{f} &\equiv \mathbf{r}_{i_{j_{1}},{j_1}} \mod{\mathbf{M}_{j_1}}, \\
      \mathbf{f} &\equiv \mathbf{r}_{i_{j_{2}},{j_2}} \mod{\mathbf{M}_{j_2}}, \\
       &\vdotswithin{\equiv}\\
      \mathbf{f} &\equiv \mathbf{r}_{i_{j_{\eta}},{j_{\eta}}} \mod{\mathbf{M}_{j_{\eta}}}.
    \end{aligned}
\end{equation}
However, if we randomly choose any $\eta$ vectors in $\mathbf{r}_{{i_1},{1}},\cdots,\mathbf{r}_{{i_{\gamma}},{\gamma}}$ and solve the congruence system (\ref{eta}), the resulting vector, which is in $\mathcal{N}(\mathbf{R}_{\{i_1,\cdots,i_{\eta}\}})$, may not be in $\mathcal{N}_{\eta}$, since $\mathcal{N}_{\eta}\subset \mathcal{N}(\mathbf{R}_{\{i_1,\cdots,i_{\eta}\}})$.
Besides, even if it is in $\mathcal{N}_{\eta}$, it may not be one of $\mathbf{f}_1,\cdots,\mathbf{f}_{\rho}$. According to Lemma \ref{lm:anotherf}, we can easily check whether it is indeed one of $\mathbf{f}_1,\cdots,\mathbf{f}_{\rho}$ by checking whether its vector remainders are in $\mathcal{S}(\mathbf{f}_1,\cdots,\mathbf{f}_{\rho})$.  

\begin{description}[leftmargin=1.26cm]
  \item[Step 2:]
  Choose any $\eta$ vectors $\mathbf{r}_{i_{j_1},j_1}, \mathbf{r}_{i_{j_2},j_2}, \cdots, \mathbf{r}_{i_{j_{\eta}},j_{\eta}}$ from the $\gamma$ chosen vectors $\mathbf{r}_{{i_1},{1}},\cdots,\mathbf{r}_{{i_{\gamma}},{\gamma}}$ in $\mathcal{S}(\mathbf{f}_1,\cdots,\mathbf{f}_{\rho})$ in Step 1 and solve the corresponding congruence system (\ref{eta}) to obtain a reconstructed vector. 
  If this newly obtained vector isn't in $\mathcal{N}_{\eta}$, we choose another $\eta$ vectors and repeat the same operations. If
  it is in $\mathcal{N}_{\eta}$, we then calculate the vector remainders $\mathbf{r}_{1},\cdots,\mathbf{r}_{\gamma}$ of this vector modulo all matrix moduli $\mathbf{M}_{j}$ for $j=1,2,\cdots,\gamma$. If $(\mathbf{r}_{1},\cdots,\mathbf{r}_{\gamma}) \in \mathcal{S}(\mathbf{f}_1,\cdots,\mathbf{f}_{\rho})$, this vector is one of $\mathbf{f}_1,\cdots,\mathbf{f}_{\rho}$, without loss of generality, denoted by $\mathbf{f}_{\rho}$. Otherwise, continue choosing another $\eta$ vectors until we find $\mathbf{f}_{\rho}$.
\end{description}

\begin{description}[leftmargin=1.26cm]
  \item[Step 3:]
  Remove the vector remainders $(\mathbf{r}_{1},\mathbf{r}_{2},\cdots,\mathbf{r}_{\gamma})$ from $\mathcal{S}(\mathbf{f}_1,\cdots,\mathbf{f}_{\rho})$ as follows:
\begin{equation}\label{delete}
\begin{aligned}
\mathcal{S}_j(\mathbf{f}_1,\cdots,\mathbf{f}_{\rho-1}) =
\begin{cases} 
\mathcal{S}_j(\mathbf{f}_1,\cdots,\mathbf{f}_{\rho}) \setminus \{\mathbf{r}_{j}\}, \\
\quad \text{if } |\mathcal{S}_j(\mathbf{f}_1,\cdots,\mathbf{f}_{\rho})| = \rho, \\
\mathcal{S}_j(\mathbf{f}_1,\cdots,\mathbf{f}_{\rho}), \quad \text{otherwise},
\end{cases}
\end{aligned}
\end{equation}
and then obtain
\begin{multline}\label{residue}
  \mathcal{S}(\mathbf{f}_1, \cdots, \mathbf{f}_{\rho-1}) = 
  S_1(\mathbf{f}_1, \cdots, \mathbf{f}_{\rho-1}) \times \cdots \\
  \times S_\gamma(\mathbf{f}_1, \cdots, \mathbf{f}_{\rho-1})
  \end{multline}
For each $j$ such that $|\mathcal{S}_j(\mathbf{f}_1,\cdots,\mathbf{f}_{\rho})|<\rho$, let $\mathcal{A}_j$ denote the subset of $\mathcal{S}_j(\mathbf{f}_1,\cdots,\mathbf{f}_{\rho})$ consisting of the vector remainders corresponding to already determined vectors, which cannot be removed from $\mathcal{S}_j(\mathbf{f}_1,\cdots,\mathbf{f}_{\rho})$ because the number of vectors in the residue set is less than $\rho$. Let $\mathcal{B}_l(j)$ denote the set of all the vector remainders that have been removed from $\mathcal{S}_j(\mathbf{f}_1,\cdots,\mathbf{f}_{\rho})$ during the past $l$ successful rounds\footnote{Although the above is for the first unknown vector $\mathbf{f}_{\rho}$ to determine, here applies to any intermediate step after $l$ vectors $\mathbf{f}_{\rho},\cdots,\mathbf{f}_{\rho-l+1}$ have been successfully determined and one could think of $\rho-l \gets \rho$.  It is similar to the set $\mathcal{A}_j$ defined above.} for $1 \leq l \leq \rho$. Then, update $\mathcal{A}_j$ and $\mathcal{B}_1(j)$.
\end{description}
 
From the above Steps 1-3, we have determined the vector $\mathbf{f}_{\rho}$ and obtained a new residue set $\mathcal{S}(\mathbf{f}_1,\cdots,$ $\mathbf{f}_{\rho-1})$. Since the number $\gamma$  of moduli is still $\eta$ times greater than the number of remaining vectors needed to be reconstructed, i.e., $\rho-1$, we can repeat Steps 1-3 over this new residue set $\mathcal{S}(\mathbf{f}_1,\cdots,\mathbf{f}_{\rho-1})$ as discussed above. Moreover, since the new residue set $\mathcal{S}(\mathbf{f}_1,\cdots,\mathbf{f}_{\rho-1})$ is the true residue set of 
$\mathbf{f}_1,\cdots,\mathbf{f}_{\rho-1}$, we can correctly determine  $\mathbf{f}_{\rho-1}$ and update $\mathcal{S}(\mathbf{f}_1,\cdots,\mathbf{f}_{\rho-2})$ by applying equations (\ref{delete}) and (\ref{residue}) with $\rho$ replaced by $\rho-1$. 

\begin{description}[leftmargin=1.26cm]
  \item[Step 4:]
  Repeat Steps 1-3 above by substituting $\mathcal{S}(\mathbf{f}_1,\cdots,\mathbf{f}_{\rho})$ with $\mathcal{S}(\mathbf{f}_1,\cdots,\mathbf{f}_{\rho-1})$ to determine $\mathbf{f}_{\rho-1}$ and update $\mathcal{S}(\mathbf{f}_1,\cdots,\mathbf{f}_{\rho-2})$, $\mathcal{A}_j$ and $\mathcal{B}_2(j)$.
\end{description}

It seems that the problem now is simply the replacement of the original problem by replacing $\rho$ with $\rho-2$, $\rho-3$, and so on, until $1$ in order to sequentially determine the vectors $\mathbf{f}_{\rho-2},\mathbf{f}_{\rho-3}, \cdots,\mathbf{f}_{1}$.
Unfortunately, the newly obtained residue set $\mathcal{S}(\mathbf{f}_1,\cdots,\mathbf{f}_{\rho-2})$ may not be the true residue set of $\mathbf{f}_1,\cdots,\mathbf{f}_{\rho-2}$.
It is because the newly obtained residue set $\mathcal{S}(\mathbf{f}_1,\cdots,\mathbf{f}_{\rho-2})$ from (\ref{delete})-(\ref{residue}) may include vector remainders in $\mathcal{A}_j$ defined in Step 3 that may not be any vector remainders of the unknown vectors $\mathbf{f}_1,\cdots,\mathbf{f}_{\rho-2}$ to determine, and has to be removed from the newly defined residue set $\mathcal{S}(\mathbf{f}_1,\cdots,\mathbf{f}_{\rho-2})$ in next steps. More details are given below.
For the convenience of the subsequent analysis, we refer each successful repetition of Steps 1-3, i.e., each time a vector is successfully determined, as a successful round. For example, Step 4 corresponds to the second successful round.

If $|\mathcal{S}_j(\mathbf{f}_1,\cdots,\mathbf{f}_{\rho})|=\rho$ for all $j=1,2,\cdots,\gamma$, i.e., the vector remainders of $\mathbf{f}_1,\cdots,\mathbf{f}_{\rho}$ modulo every single matrix are distinct, then we can always get the correct residue set $\mathcal{S}(\mathbf{f}_1,\cdots,\mathbf{f}_{i})$ that is the true residue set of $\mathbf{f}_1,\cdots,\mathbf{f}_{i}$, for all $i$ with $1 \leq i \leq \rho-1$.
Therefore, no error will occur during the repetition of Steps 1-3 in any repetition, i.e., we can always find a vector satisfying the required conditions at each repetition, which is one of the vectors that we want to determine. This case is shown in Example \ref{ex:1} below.

However, if there exists some $1 \leq j \leq \gamma$ such that  $|\mathcal{S}_j(\mathbf{f}_1,\cdots,\mathbf{f}_{\rho})|=\alpha<\rho$, it is possible that during the $(\rho-\alpha+1)$-th repetition of Steps 1-3, a vector remainder that is still required for determining some undetermined vectors is mistakenly removed according to equation (\ref{delete}). 
This is because, during the first $(\rho-\alpha)$ repetitions of Steps 1-3, there is no need to remove vector remainders from $\mathcal{S}_j(\mathbf{f}_1,\cdots,\mathbf{f}_{\rho})$. However, in the $(\rho-\alpha+1)$-th repetition, we need to remove the vector remainder $\mathbf{r}_{\alpha,j}$ of the newly reconstructed vector $\mathbf{f}_{\alpha}$ modulo $\mathbf{M}_j$ from the set $\mathcal{S}_j(\mathbf{f}_1,\cdots,\mathbf{f}_{\alpha})$. This removed vector, however, may be the same vector remainder of two or more different vectors $\mathbf{f}_i$, and thus may actually be the vector remainder of one of the yet-to-be-determined vectors $\mathbf{f}_1,\cdots,\mathbf{f}_{\alpha-1}$. In such a case, the vector that should have been removed is one from the set $\{\mathbf{r}_{\rho,j},\mathbf{r}_{\rho-1,j},\cdots,\mathbf{r}_{\alpha+1,j}\}$, i.e., $\mathcal{A}_j$ defined in Step 3 earlier, which contains the remainders of all the previously reconstructed vectors  $\mathbf{f}_{\rho},\mathbf{f}_{\rho-1},\cdots,\mathbf{f}_{\alpha+1}$  modulo $\mathbf{M}_{j}$. This case is shown in Example \ref{ex:3.2} below.

According to Lemma \ref{lm:anotherf}, if a vector remainder is mistakenly removed in the above sense, then in one of the subsequent repetitions of Steps 1-3, it will become impossible to find a vector satisfying the requirements of Step 2.

Assume that we fail to get a vector in the repetition of Steps 1-3 after $l$ successful rounds, where $l > \rho-\alpha$. This indicates that the earlier residue set $\mathcal{S}_{j}(\mathbf{f}_1,\cdots,\mathbf{f}_{\rho-l})$ is incorrect due to the removal of a necessary vector remainder. To address this, we construct a new residue set by removing one vector from the intersection $\mathcal{S}_{j}(\mathbf{f}_1,\cdots,\mathbf{f}_{\rho-l}) \cap \mathcal{A}_j$ and adding back one of the vector remainders in set $\mathcal{B}_{l}(j)$, which is the set of all the vectors that are removed from $\mathcal{S}(\mathbf{f}_1,\cdots,\mathbf{f}_{\rho})$ in the last $l$ successful rounds. We then repeat Steps 1-3 using this updated residue set.
If the replacement is correct, i.e., the residue set is the correct residue set of the remaining vectors $\mathbf{f}_1,\cdots,\mathbf{f}_{\rho-l}$ to determine,
the subsequent repetitions of Steps 1-3 will proceed without error, and we can successfully reconstruct all vectors after $\rho-l$ repetitions. If the replacement is incorrect, an error will arise again in a later repetition, without loss of generality, say, after the $l^{\prime}$-th successful round. We then update $\mathcal{S}_{j}(\mathbf{f}_1,\cdots,\mathbf{f}_{\rho-l^{\prime}})$ by removing one vector from the intersection $\mathcal{S}_{j}(\mathbf{f}_1,\cdots,\mathbf{f}_{\rho-l^{\prime}}) \cap \mathcal{A}_j$ and adding back one of the vector remainders in set $\mathcal{B}_{l^{\prime}}(j)$ and repeat Steps 1-3.
Since the numbers of vectors in $\mathcal{A}_j$ and $\mathcal{B}_l(j)$ for $l > \rho-\alpha$, are finite, we are guaranteed to eventually find the correct residue set. As a result, all vectors $\mathbf{f}_1,\cdots,\mathbf{f}_{\rho}$ can be successfully determined.

For $5\leq k \leq \rho+2$, we have following steps:

\begin{description}[leftmargin=1.26cm]
  \item[Step k:]
  Repeat Steps 1-2. If a vector satisfying the required conditions is found, proceed to Step 3. Otherwise, update the residue set by removing one vector from the intersection $$\mathcal{S}_{j}(\mathbf{f}_1,\cdots,\mathbf{f}_{\rho-k+3}) \cap \mathcal{A}_j$$ and adding back one of the vector remainders in set $\mathcal{B}_{k-3}(j)$. Then repeat Steps 1-2 using the updated residue set. If it still fails, try a different candidate for removal and repeat the process until a valid vector is found. Once such a vector is obtained, proceed to Step 3.
\end{description}

By following the above procedure, we are guaranteed to successfully determine all $\rho$ vectors $\mathbf{f}_1,\cdots,\mathbf{f}_{\rho}$. If there are more than one vector residue sets with their cardinality less than $\rho$, we apply the above correction operations progressively. We first attempt to correct each vector residue set individually. If all individual attempts fail, we then try all possible combinations of two vector residue sets, followed by combinations of three, and so on.

Let's briefly analyze the complexity of this algorithm when $|\mathcal{S}_j(\mathbf{f}_1,\cdots,\mathbf{f}_{\rho})|=\rho$ for all $j=1,2,\cdots,\gamma$. 
For each $1 \leq l \leq \rho$, in the $l$-th successful round, the chosen vector remainders $(\mathbf{r}_{{i_1},{1}},\cdots,\mathbf{r}_{{i_{\gamma}},{\gamma}})$ can be partitioned into $\rho-l+1$ disjoint groups, each corresponding to a single vector. If $\alpha_l$ groups contain $\eta_l+1$ vector remainders and the remaining $(\rho-l+1-\alpha_l)$ groups contain $\eta_l$ vector remainders, where $\eta_l$ and $\alpha_l$ satisfy $\gamma=\eta_l(\rho-l+1)+\alpha_l$ and $0\leq \alpha_l < \rho-l+1$, we need to try the maximal possible combinations of $\eta$ vector remainders in order to find a valid vector, leading to the worst-case number of MD-CRT to solve in this round. Therefore, in the $l$-th round, we need to use MD-CRT at most 
$$\binom{\gamma}{\eta}-\alpha_l\binom{\eta_l+1}{\eta}-(\rho-l+1-\alpha_l)\binom{\eta_l}{\eta}+1$$
times.
Combining all $\rho$ rounds, totally we need to use MD-CRT at most 
\begin{equation}\label{complexity}
    \sum_{l=1}^{\rho}\left[\binom{\gamma}{\eta}-\alpha_l\binom{\eta_l+1}{\eta}-(\rho-l+1-\alpha_l)\binom{\eta_l}{\eta}+1\right]
\end{equation}
times. Note that $\eta_1=\eta$. When there exist $j$ such that $|\mathcal{S}_j(\mathbf{f}_1,\cdots,\mathbf{f}_{\rho})| \neq \rho$, i.e., it has repeated remainders, the above analysis shows that the complexity is in the same order as (\ref{complexity}) and furthermore, the probability of having repeated remainders is small in practical applications.

A detailed algorithm is summarized in Algorithm \ref{alg:reconstruction}, which, for simplicity, only focuses on the case that at most one vector residue set with its cardinality less than $\rho$. Next, we present two examples to show how this algorithm works.

\begin{algorithm}[htbp]
\caption{Reconstruction of $\rho$ Integer Vectors without Prior Information}
\label{alg:reconstruction}
\begin{algorithmic}[1]
\Require Residue set $\mathcal{S}=\mathcal{S}_1 \times \cdots \times \mathcal{S}_{\gamma}$ with $\gamma$ matrix moduli $\mathbf{M}_1,\cdots,\mathbf{M}_{\gamma}$ and $\mathcal{N}_{\eta}$
\Ensure $\rho$ integer vectors $\mathbf{f}_1, \cdots, \mathbf{f}_\rho$

\State Initialize $k \gets 1$, successful\_rounds $\gets 0$
\While{$k \leq \rho$}
    \State \textbf{Step 1}: Arbitrarily choose $\gamma$ vector remainders $(\mathbf{r}_{i_1,1}, \cdots, \mathbf{r}_{i_\gamma,\gamma})$ from current residue set $\mathcal{S}$
    \State \textbf{Step 2}: \For{each combination of $\eta$ remainders from the chosen $\gamma$ vector remainders}
    \State Solve congruence system (\ref{eta}) to get $\mathbf{f}$ 
    \If{$\mathbf{f} \in \mathcal{N}_{\eta}$ and its remainders $(\mathbf{r}_1, \cdots, \mathbf{r}_\gamma) \in \mathcal{S}$ }
        \State \textbf{Step 3}: Set $\mathbf{f}_{\rho - k + 1} \gets \mathbf{f}$
        \State Remove remainders $(\mathbf{r}_1, \cdots, \mathbf{r}_\gamma)$ as follows:
        \State \parbox[t]{.7\linewidth}{
            $\mathcal{S}_j' =
              \begin{cases}
             \mathcal{S}_j \setminus \{\mathbf{r}_j\}, & \text{if } 
             |\mathcal{S}_j| = \rho - k + 1, \\
            \mathcal{S}_j, & \text{otherwise}.
            \end{cases}$
         }
        \State Update residue set: $\mathcal{S} \gets \mathcal{S}_1' \times \cdots \times \mathcal{S}_\gamma'$
        \State Update $\mathcal{A}_j$ and $\mathcal{B}_{\text{successful\_rounds}+1}(j)$
        \State successful\_rounds $\gets$ successful\_rounds $+ 1$
        \State $k \gets k + 1$
        \State break
        \EndIf
    \EndFor
\If{no valid vector found}
    \State Remove one vector from $\mathcal{S}_j \cap \mathcal{A}_j$
    \State Add back one from $\mathcal{B}_{k-1}(j)$
    \State Repeat Steps 1-2 with updated residue set
\EndIf
\EndWhile
\State \Return $\mathbf{f}_1, \ldots, \mathbf{f}_\rho$
\end{algorithmic}
\end{algorithm}

\begin{example}\label{ex:1}
    Let $\rho=2$ and $\gamma=4$. The four matrix moduli are as follows:
    \begin{equation}\notag
\begin{aligned} 
\mathbf{M}_{1} &=
    \begin{pmatrix}
        3 & 0 \\
        1 & 3 \\
    \end{pmatrix},
    \mathbf{M}_{2} =
    \begin{pmatrix}
        3 & 1 \\
        0 & 3 \\
    \end{pmatrix},\\ 
    \mathbf{M}_{3} &=
    \begin{pmatrix}
        4 & 0 \\
        1 & 4 \\
    \end{pmatrix},
    \mathbf{M}_{4} =
    \begin{pmatrix}
        4 & 1 \\
        0 & 4 \\
    \end{pmatrix}.
\end{aligned}
\end{equation}
Since $\eta=2$, we calculate all $6$ lcrms of any two matrices of $\mathbf{M}_1,\mathbf{M}_2,\mathbf{M}_3,\mathbf{M}_4$ following \cite{MD2,remainder} and get
\begin{equation}\notag
\begin{aligned} 
\mathbf{R}_{\{1,2\}} &=
    \begin{pmatrix}
        9 & 0 \\
        0 & 9 \\
    \end{pmatrix},
    \mathbf{R}_{\{1,3\}} =
    \begin{pmatrix}
        12 & 0 \\
        -5 & 12 \\
    \end{pmatrix},\\ 
    \mathbf{R}_{\{1,4\}} &=
    \begin{pmatrix}
        3 & 0 \\
        -20 & 48 \\
    \end{pmatrix},
    \mathbf{R}_{\{2,3\}} =
    \begin{pmatrix}
        4 & 0 \\
        -15 & 36 \\
    \end{pmatrix},\\
    \mathbf{R}_{\{2,4\}} &=
    \begin{pmatrix}
        12 & -5 \\
        0 & 12 \\
    \end{pmatrix},
    \mathbf{R}_{\{3,4\}} =
    \begin{pmatrix}
        16 & 0 \\
        0 & 16 \\
    \end{pmatrix}.
\end{aligned}
\end{equation}
\begin{figure}[htbp]
    \centering
    \includegraphics[width=\columnwidth]{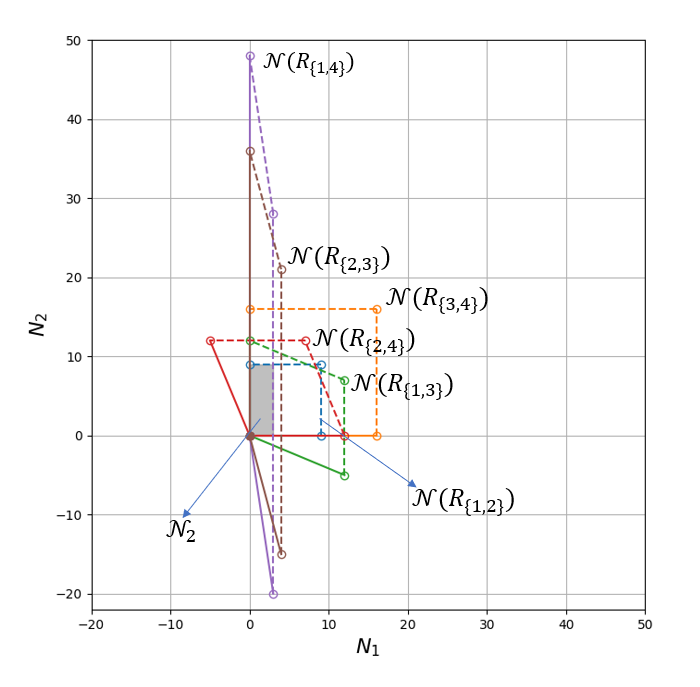}  
    \caption{FPDs of the six lcrms in Example \ref{ex:1}}
    \label{fig:ex1}  
\end{figure}
The FPDs of all these $6$ matrices are shown in Fig. \ref{fig:ex1}, where the shaded area represents their intersection, i.e., $\mathcal{N}_{2}$. It is not hard to see that $\mathcal{N}_{2}=\mathcal{N}(\begin{pmatrix}
    3 &0\\
    0 &9\\
\end{pmatrix})$. Assume the two vectors are $\mathbf{f}_1=[2,4]^{\top}$ and $\mathbf{f}_{2}=[1,7]^{\top}$. We then calculate the residue sets 
\begin{equation*}
\begin{aligned}
    \mathcal{S}_1(\mathbf{f}_1,\mathbf{f}_2) &= \{[2,1]^{\top},[1,1]^{\top}\},\\
    \mathcal{S}_2(\mathbf{f}_1,\mathbf{f}_2) &= \{[1,1]^{\top},[2,1]^{\top}\},\\
    \mathcal{S}_3(\mathbf{f}_1,\mathbf{f}_2) &= \{[2,4]^{\top},[1,3]^{\top}\},\\
    \mathcal{S}_4(\mathbf{f}_1,\mathbf{f}_2) &= \{[1,0]^{\top},[4,3]^{\top}\}.\\
\end{aligned}    
\end{equation*}
We choose $([2,1]^{\top},[2,1]^{\top},[1,3]^{\top},[1,0]^{\top})$ from $\mathcal{S}(\mathbf{f}_1,\mathbf{f}_2)$. We first choose $\mathbf{r}_{1}=[2,1]^{\top}$ and $\mathbf{r}_{2}=[2,1]^{\top}$ and solve 
\begin{equation*}
    \begin{aligned}
      \mathbf{f} &\equiv \mathbf{r}_{1} \mod{\mathbf{M}_{1}}, \\
      \mathbf{f} &\equiv \mathbf{r}_{2} \mod{\mathbf{M}_{2}}, 
    \end{aligned}
\end{equation*}
to get $\mathbf{f}=[2,1]^{\top} \in \mathcal{N}_{2}$. But its vector remainders $([2,1]^{\top},[2,1]^{\top},[2,1]^{\top},[2,1]^{\top}) \notin \mathcal{S}(\mathbf{f}_1,\mathbf{f}_2)$. So we continue choose $\mathbf{r}_{2}=[2,1]^{\top}$ and $\mathbf{r}_{3}=[1,3]^{\top}$ and then solve a similar congruence systems to get $\mathbf{f}=[1,7]^{\top}\in \mathcal{N}_{2}$. We can check that its vector remainders $([1,1]^{\top},[2,1]^{\top},[1,3]^{\top},[4,3]^{\top}) \in \mathcal{S}(\mathbf{f}_1,\mathbf{f}_2)$, which means $[1,7]^{\top}$ is what we want and denoted by $\mathbf{f}_2$. Then, we delete its vector remainders from $\mathcal{S}(\mathbf{f}_1,\mathbf{f}_2)$ and get
$\mathcal{S}(\mathbf{f}_1)=\{[2,1]^{\top}\} \times \{[1,1]^{\top}\} \times \{[2,4]^{\top}\} \times \{[1,0]^{\top}\}$. And we can use MD-CRT to get $\mathbf{f}_1=[2,4]^{\top}$.
\end{example}

In this example, one can see that the maximal dynamic range of any  matrix modulus $\mathbf{M}_j$ alone is 16, i.e., if only one of these matrix moduli is used, any 2 of 16 possible integer vectors can be uniquely determined. However, if 4 matrix moduli are used with our proposed algorithm, the dynamic range is at least 27. In other words, any 2 of 27 possible integer vectors can be uniquely determined.

\begin{example}\label{ex:3.2}
    Let $\rho=3$ and $\gamma=6$. $\mathbf{M}_1,\mathbf{M}_2,\mathbf{M}_3,\mathbf{M}_4$ are the same with Example \ref{ex:1} and
    \begin{equation}\notag
\begin{aligned} 
    \mathbf{M}_{5} &=
    \begin{pmatrix}
        5 & 0 \\
        1 & 5 \\
    \end{pmatrix},
    \mathbf{M}_{6} =
    \begin{pmatrix}
        5 & 1 \\
        0 & 5 \\
    \end{pmatrix}.
\end{aligned}
\end{equation}
Since $\eta=2$, we need to calculate all $15$ lcrms of any two matrices of $\mathbf{M}_1,\mathbf{M}_2,\mathbf{M}_3,\mathbf{M}_4,\mathbf{M}_5,\mathbf{M}_6$ following \cite{MD2,remainder}. If we choose the same $\mathbf{R}_{\{1,2\}}$, $\mathbf{R}_{\{1,3\}}$, $\mathbf{R}_{\{1,4\}}$, $\mathbf{R}_{\{2,3\}}$, $\mathbf{R}_{\{2,4\}}$, $\mathbf{R}_{\{3,4\}}$ with Example \ref{ex:1} and appropriately select the other nine lcrms, we can get the same $\mathcal{N}_{2}$ with Example \ref{ex:1}, i.e.,  
$\mathcal{N}_{2}=\mathcal{N}(\begin{pmatrix}
    3 &0\\
    0 &9\\
\end{pmatrix})$. Assume the three vectors are $\mathbf{f}_1=[2,6]^{\top}$, $\mathbf{f}_2=[1,8]^{\top}$ and $\mathbf{f}_{3}=[0,3]^{\top}$. We then calculate the residue sets 
\begin{equation*}
\begin{aligned}
    \mathcal{S}_1(\mathbf{f}_1,\mathbf{f}_2,\mathbf{f}_3) &= \{[2,3]^{\top},[1,2]^{\top}, [0,0]^{\top}\},\\
    \mathcal{S}_2(\mathbf{f}_1,\mathbf{f}_2,\mathbf{f}_3) &= \{[0,0]^{\top},[2,2]^{\top}, [2,0]^{\top}\},\\
    \mathcal{S}_3(\mathbf{f}_1,\mathbf{f}_2,\mathbf{f}_3) &= \{[2,2]^{\top},[1,4]^{\top}, [0,3]^{\top}\},\\
    \mathcal{S}_4(\mathbf{f}_1,\mathbf{f}_2,\mathbf{f}_3) &= \{[1,2]^{\top},[3,0]^{\top}, [4,3]^{\top}\},\\
    \mathcal{S}_5(\mathbf{f}_1,\mathbf{f}_2,\mathbf{f}_3) &= \{[2,1]^{\top},[1,3]^{\top}, [0,3]^{\top}\},\\
    \mathcal{S}_6(\mathbf{f}_1,\mathbf{f}_2,\mathbf{f}_3) &= \{[1,1]^{\top},[5,3]^{\top}\}.\\
\end{aligned}    
\end{equation*}
Choose $([1,2]^{\top},[0,0]^{\top},[1,4]^{\top},[4,3]^{\top},[1,3]^{\top},[1,1]^{\top})$ from $\mathcal{S}(\mathbf{f}_1,\mathbf{f}_2,\mathbf{f}_3)$. We first choose $\mathbf{r}_{2}=[0,0]^{\top}$, $\mathbf{r}_{3}=[1,4]^{\top}$ and solve 
\begin{equation}\label{ex:cong}
    \begin{aligned}
      \mathbf{f} &\equiv \mathbf{r}_{2} \mod{\mathbf{M}_{2}}, \\
      \mathbf{f} &\equiv \mathbf{r}_{3} \mod{\mathbf{M}_{3}}, 
    \end{aligned}
\end{equation}
to get $\mathbf{f}=[0,24]^{\top} \notin \mathcal{N}_{2}$. So we continue choose $\mathbf{r}_{2}=[0,0]^{\top}$ and $\mathbf{r}_{6}=[1,1]^{\top}$ and then solve a similar congruence system in (\ref{ex:cong}) to get $\mathbf{f}=[2,6]^{\top}\in \mathcal{N}_{2}$. We can check that its vector remainders $$([2,3]^{\top},[0,0]^{\top},[2,2]^{\top},[1,2]^{\top},[2,1]^{\top},[1,1]^{\top}) \in \mathcal{S}(\mathbf{f}_1,\mathbf{f}_2,\mathbf{f}_3),$$ which means $[2,6]^{\top}$ is what we want and denoted by $\mathbf{f}_3$.

Then, we remove the remainders of $\mathbf{f}_3$ from the residue sets and get the residue sets of $\mathbf{f}_1$ and $\mathbf{f}_2$:
\begin{equation*}
\begin{aligned}
    \mathcal{S}_1(\mathbf{f}_1,\mathbf{f}_2) &= \{[1,2]^{\top}, [0,0]^{\top}\},\\
    \mathcal{S}_2(\mathbf{f}_1,\mathbf{f}_2) &= \{[2,2]^{\top}, [2,0]^{\top}\},\\
    \mathcal{S}_3(\mathbf{f}_1,\mathbf{f}_2) &= \{[1,4]^{\top}, [0,3]^{\top}\},\\
    \mathcal{S}_4(\mathbf{f}_1,\mathbf{f}_2) &= \{[3,0]^{\top}, [4,3]^{\top}\},\\
    \mathcal{S}_5(\mathbf{f}_1,\mathbf{f}_2) &= \{[1,3]^{\top}, [0,3]^{\top}\},\\
    \mathcal{S}_6(\mathbf{f}_1,\mathbf{f}_2) &= \{[1,1]^{\top},[5,3]^{\top}\}.\\
\end{aligned}    
\end{equation*}
Since $|\mathcal{S}_6(\mathbf{f}_1,\mathbf{f}_2,\mathbf{f}_3)|<3$, we calculate $\mathcal{A}_6=\{[1,1]^{\top}\}$ and $\mathcal{B}_1(6)=\varnothing$.

We continue to choose $([1,2]^{\top},[2,2]^{\top},[0,3]^{\top},[3,0]^{\top},$ $[0,3]^{\top},[5,3]^{\top})$ from $\mathcal{S}(\mathbf{f}_1,\mathbf{f}_2)$ and then choose $\mathbf{r}_{5}=[0,3]^{\top}$ and $\mathbf{r}_{6}=[5,3]^{\top}$ and then solve the corresponding congruence system to get $\mathbf{f}_2=[0,3]^{\top}\in \mathcal{N}_{2}$ and its vector remainders are in $\mathcal{S}(\mathbf{f}_1,\mathbf{f}_2)$.
Then, we delete its vector remainders from $\mathcal{S}(\mathbf{f}_1,\mathbf{f}_2)$ and get
$\mathcal{S}(\mathbf{f}_1)=\{[1,2]^{\top}\} \times \{[2,2]^{\top}\} \times \{[1,4]^{\top}\} \times \{[3,0]^{\top}\} \times \{[1,3]^{\top}\} \times \{[1,1]^{\top}\}$, $\mathcal{A}_6=\{[1,1]^{\top}\}$ and $\mathcal{B}_2(6)=\{[5,3]^{\top}\}$. In fact, the undetermined vector $\mathbf{f}_1$ shares the same vector remainder $[5,3]^{\top}$ with $\mathbf{f}_2$ when modulo $\mathbf{M}_6$, but this vector was removed after $\mathbf{f}_2$ was correctly determined in the earlier steps, while the wrong vector $[1,1]^{\top}$ was retained in the newly obtained reside set $\mathcal{S}(\mathbf{f}_1)$. As a result, an error occurs in the subsequent steps and we need to use $\mathcal{A}_6$ and $\mathcal{B}_2(6)$ to recover the true residue set $\mathcal{S}(\mathbf{f}_1)$ of $\mathbf{f}_1$ as introduced in the above algorithm.

Now, we can choose any two vectors in $\mathcal{S}(\mathbf{f}_1)$ and solve their corresponding congruence system. We either get a vector that doesn't belong to the set $\mathcal{N}_{2}$, or it belongs to $\mathcal{N}_{2}$ but the tuple of vector remainders doesn't belong to $\mathcal{S}(\mathbf{f}_1)$. There is a simpler way by using MD-CRT directly on these six vector remainders to uniquely get a vector in $\mathcal{N}(\mathbf{R}_{\{1,2,3,4,5,6\}})$, where $\mathbf{R}_{\{1,2,3,4,5,6\}}=diag\{3600,3600\}$, since there is only one vector need to determine and $\mathcal{N}_2 \subset \mathcal{N}(\mathbf{R}_{\{1,2,3,4,5,6\}})$. 
By calculating, we get $\mathbf{f}=[1441,3176]^{\top} \notin \mathcal{N}_{2}$, which means $\mathcal{S}(\mathbf{f}_1)$ is incorrect. Therefore, we remove $[1,1]^{\top}\in \mathcal{A}_6$ from $\mathcal{S}_6(\mathbf{f}_1)$ and add $[5,3]^{\top}\in \mathcal{B}_2(6)$ back. In case there were more than 1 element in $\mathcal{A}_6$, we would try one element by one element in $\mathcal{A}_6$ to remove. If for a removed element, an error still occurs, we try another element in $\mathcal{A}_6$ to remove until there is no error occurs. This analysis can be readily applied to the case there were more than 1 element in $\mathcal{B}_2(6)$. Then, we use MD-CRT again and get $\mathbf{f}_1=[1,8]^{\top} \in \mathcal{N}_{2}$. 
\end{example}

In the above two examples, we have chosen the lcrms as described in order to make the structure of $\mathcal{N}_2$ clear to view. However, the resulting $\mathcal{N}_2$ does not have the maximal possible size, although it is larger than the absolute determinant values of all above six matrix moduli. This is because the lcrm of any group of integer matrices is not unique and thus the intersection ${\cal N}_2$ may have larger sizes, i.e., its dynamic range may be larger, as shown below. A larger $\mathcal{N}_2$ than that in Fig. \ref{fig:ex1} can be obtained by choosing  
\begin{equation}\notag
\begin{aligned} 
    \mathbf{R}_{\{1,4\}} =
    \begin{pmatrix}
        3 & 3 \\
        -20 & 28 \\
    \end{pmatrix},
    \mathbf{R}_{\{2,3\}} =
    \begin{pmatrix}
        4 & 4 \\
        -15 & 21 \\
    \end{pmatrix},
\end{aligned}
\end{equation}
which also means a larger dynamic range. However, the maximal possible size of $\mathcal{N}_{\eta}$ for a given group of $\gamma$ matrix moduli remains unknown.

The above algorithm is an extension of that in \cite{xia00} from scalar integers to integer vectors. Although it is similar, the work in this current paper completes the algorithm in \cite{xia00} by addressing the potential failures that were previously overlooked and also by providing a complexity analysis that is missing in \cite{xia00}.

\section{MD-CRT for Two Integer Vectors with Prior Information}\label{s4}

In this section, we focus on the second question, Question \ref{q2}, that we have proposed to study in this paper, i.e., if some prior information or some condition is given for the unknown integer vectors $\mathbf{f}_1,\cdots,\mathbf{f}_{\rho}$, can the maximal possible dynamic range be achieved?
We next only consider this question for two integer vectors, i.e., $\rho=2$. A new condition that achieves the maximal possible dynamic range is proposed with a detailed algorithm. Furthermore, we show that our newly obtained condition achieves a better result than all the existing ones when it returns to one dimension.

\subsection{A Condition to Achieve the Maximal Possible Dynamic Range}

Before presenting the main result, we introduce three lemmas.

\begin{lemma}\label{lm:1}
    Consider two integer vectors $\mathbf{f}_1$ and $\mathbf{f}_2$, and $\gamma$ non-singular integer matrices $\mathbf{M}_j$, $j=1,2,\cdots,\gamma$. Let $\mathbf{r}_{1,j}$ and $\mathbf{r}_{2,j}$ be the vector remainders of $\mathbf{f}_1$ and $\mathbf{f}_2$ modulo $\mathbf{M}_j$, respectively, for $j=1,2,\cdots,\gamma$. Let $\mathbf{R}$ be an lcrm of all the matrices $\mathbf{M}_j$, $j=1,2,\cdots,\gamma$. Then, for any $1 \leq j_1\neq j_2 \leq \gamma$, the vector remainder of $\mathbf{r}_{1,j_1} - \mathbf{r}_{2,j_1}$ modulo $\mathbf{M}_{j_1}$ is equal to that of $\mathbf{r}_{1,j_2} - \mathbf{r}_{2,j_2}$ modulo $\mathbf{M}_{j_2}$ if and only if
    \begin{equation}\label{eq:lm4.1}
       \mathbf{f}_1 - \mathbf{f}_2 \in LAT(\mathbf{R}) + \bigcap\limits_{j=1}^{\gamma}\mathcal{N}(\mathbf{M}_j). 
    \end{equation} 
\end{lemma}

\begin{proof}
    For $\mathbf{f}_1$ and $\mathbf{f}_2$, and their vector remainders $\mathbf{r}_{1,j}$ and $\mathbf{r}_{2,j}$, there exist unique integer vectors $\mathbf{n}_{1,j}$ and $\mathbf{n}_{2,j}$ for each $j=1,2,\cdots,\gamma$, such that 
    \begin{equation}\label{f1f2}
       \mathbf{f}_1=\mathbf{M}_j \mathbf{n}_{1,j}+\mathbf{r}_{1,j} \quad \text{and} \quad \mathbf{f}_2=\mathbf{M}_j \mathbf{n}_{2,j}+\mathbf{r}_{2,j}.  
    \end{equation}
    Then, we can get 
    \begin{equation}\notag
        \mathbf{f}_1-\mathbf{f}_2 \equiv \mathbf{r}_{1,j} - \mathbf{r}_{2,j} \mod \mathbf{M}_j,
    \end{equation}
    which means that the vector remainder of $\mathbf{r}_{1,j} - \mathbf{r}_{2,j}$ modulo $\mathbf{M}_j$ is equal to the vector remainder of $\mathbf{f}_1-\mathbf{f}_2$ modulo $\mathbf{M}_j$.
    As a result, we just need to prove that for any given $1 \leq j_1\neq j_2 \leq \gamma$, the vector remainder of $\mathbf{f}_1-\mathbf{f}_2$ modulo $\mathbf{M}_{j_1}$ is equal to the vector remainder of $\mathbf{f}_1-\mathbf{f}_2$ modulo $\mathbf{M}_{j_2}$ if and only if (\ref{eq:lm4.1}) holds.

    For the sufficient part,
    since $$\mathbf{f}_1 - \mathbf{f}_2 \in LAT(\mathbf{R}) + \bigcap\limits_{j=1}^{\gamma}\mathcal{N}(\mathbf{M}_j),$$ there exist two integer vectors $$\mathbf{n} \in LAT(\mathbf{R}) \quad \text{and} \quad \mathbf{r} \in \bigcap\limits_{j=1}^{\gamma}\mathcal{N}(\mathbf{M}_j)$$ such that $\mathbf{f}_1 - \mathbf{f}_2 = \mathbf{n} + \mathbf{r}$. From Prop. \ref{caplattice}, we have 
    $$LAT(\mathbf{R}) = \bigcap\limits_{j=1}^{\gamma} LAT(\mathbf{M}_j),$$
    which means that there exist two integer vectors $\mathbf{n}_1$ and $\mathbf{n}_2$ such that $\mathbf{n}=\mathbf{M}_{j_1}\mathbf{n_1}$ and $\mathbf{n}=\mathbf{M}_{j_2}\mathbf{n_2}$.
    Then, we have $\mathbf{f}_1-\mathbf{f}_2=\mathbf{M}_{j_1}\mathbf{n}_{1}+\mathbf{r}$ and $\mathbf{f}_1-\mathbf{f}_2=\mathbf{M}_{j_2}\mathbf{n}_{2}+\mathbf{r}$. Since $\mathbf{r} \in \mathcal{N}(\mathbf{M}_{j_1}) \cap \mathcal{N}(\mathbf{M}_{j_2})$, we have $\langle \mathbf{f}_1-\mathbf{f}_2 \rangle_{\mathbf{M}_{j_1}} = \langle \mathbf{f}_1-\mathbf{f}_2 \rangle_{\mathbf{M}_{j_2}} = \mathbf{r}$.

    For the necessary part, the vector remainders of $\mathbf{f}_{1}-\mathbf{f}_2$ modulo all the matrices $\mathbf{M}_j$, $j=1,2,\cdots,\gamma$, are equal, denoted by $\mathbf{r}$, and $\mathbf{r}\in \bigcap\limits_{j=1}^{\gamma}\mathcal{N}(\mathbf{M}_j)$. 
    In addition, $\mathbf{f}_1-\mathbf{f}_2$ has unique representation 
    $$\mathbf{f}_1-\mathbf{f}_2=\mathbf{M}_j\mathbf{n}_j+\mathbf{r},$$
    for each $j=1,2,\cdots,\gamma$, where $\mathbf{n}_j$ is an integer vector.

    As a result, 
    \begin{align*}
    \mathbf{f}_1 - \mathbf{f}_2 
     \in \bigcap\limits_{j=r}^{\gamma} \left( LAT(\mathbf{M}_j) + \mathbf{r} \right) 
     &\subset LAT(\mathbf{R}) + \mathbf{r} \\
     &\subset LAT(\mathbf{R}) + \bigcap\limits_{j=1}^{\gamma} \mathcal{N}(\mathbf{M}_j).
\end{align*}
    This completes the proof.
\end{proof}

\begin{lemma}\label{lm:2}
   Consider two integer vectors $\mathbf{f}_1$ and $\mathbf{f}_2$, and $\gamma$ non-singular integer matrices $\mathbf{M}_j$, $j=1,2,\cdots,\gamma$. Let $\mathbf{r}_{1,j}$ and $\mathbf{r}_{2,j}$ be the vector remainders of $\mathbf{f}_1$ and $\mathbf{f}_2$ modulo $\mathbf{M}_j$, respectively, for $j=1,2,\cdots,\gamma$. Let $\mathbf{R}$ be an lcrm of all the matrices $\mathbf{M}_j$, $j=1,2,\cdots,\gamma$. If 
    $$\mathbf{f}_1 - \mathbf{f}_2 \in LAT(\mathbf{R}) +\Big{(} \bigcap\limits_{j=1}^{\gamma}\mathcal{N}(\mathbf{M}_j)\setminus \{\mathbf{0}\}\Big{)} ,$$ 
    then  $\mathbf{r}_{1,j}\neq \mathbf{r}_{2,j}$ holds for all $j=1,2,\cdots,\gamma$. 
\end{lemma}

\begin{proof}
    If there exists some $j$ such that $\mathbf{r}_{1,j}= \mathbf{r}_{2,j}$, we can get that $\mathbf{f}_1 - \mathbf{f}_2 = \mathbf{M}_j(\mathbf{n}_{1,j}-\mathbf{n}_{2,j})$ from (\ref{f1f2}), which means that $\mathbf{f}_1 - \mathbf{f}_2 \in LAT(\mathbf{M}_j)$. So, it is impossible that $\mathbf{f}_1 - \mathbf{f}_2 \in LAT(\mathbf{M}_j) + (\mathcal{N}(\mathbf{M}_j)\setminus \{\mathbf{0}\})$.

    Since
    \begin{align*}
LAT(\mathbf{R}) +\Big{(} \bigcap\limits_{j=1}^{\gamma}\mathcal{N}(\mathbf{M}_j)&\setminus \{\mathbf{0}\}\Big{)}\\  &\subset LAT(\mathbf{M}_j) + \Big{(}\mathcal{N}(\mathbf{M}_j)\setminus \{\mathbf{0}\}\Big{)},
    \end{align*}
    it is impossible that 
    $$\mathbf{f}_1 - \mathbf{f}_2 \in LAT(\mathbf{R}) +\Big{(} \bigcap\limits_{j=1}^{\gamma}\mathcal{N}(\mathbf{M}_j)\setminus \{\mathbf{0}\}\Big{)} ,$$ 
    which contradicts the condition.
\end{proof}

For any integer matrix $\mathbf{M}$, denote $$\mathcal{L}_{1/2}(\mathbf{M}) = \{ \mathbf{k}\in \mathbb{Z}^{D}\ |\ \mathbf{k}=\mathbf{M}(\frac{1}{2}\mathbf{n}) \ \text{for some}\ \mathbf{n} \in \mathbb{Z}^{D}\}.$$ 
 We can easily see that $$\mathcal{L}_{1/2}(\mathbf{M}) = \frac{1}{2}LAT(\mathbf{M}) \cap \mathbb{Z}^{D}.$$
In other words, $\mathcal{L}_{1/2}(\mathbf{M})$ is the set of all the integer vectors in $1/2$ scaled lattice generated by $\mathbf{M}$. And all the vectors in $\mathcal{L}_{1/2}(\mathbf{M})$ are symmetric in terms of the zero vector, i.e., if $\mathbf{k} \in \mathcal{L}_{1/2}(\mathbf{M})$, we also have $-\mathbf{k} \in \mathcal{L}_{1/2}(\mathbf{M})$.
Then, we have the following lemma.

\begin{lemma}\label{lm:3}
    Consider two integer vectors $\mathbf{f}_1$ and $\mathbf{f}_2$, and $\gamma$ non-singular integer matrices $\mathbf{M}_j$, $j=1,2,\cdots,\gamma$. Let $\mathbf{r}_{1,j}$ and $\mathbf{r}_{2,j}$ be the vector remainders of $\mathbf{f}_1$ and $\mathbf{f}_2$ modulo $\mathbf{M}_j$, respectively, for $j=1,2,\cdots,\gamma$. Then, $$\mathbf{r}_{1,j} - \mathbf{r}_{2,j} \not\equiv \mathbf{r}_{2,j} - \mathbf{r}_{1,j} \mod \mathbf{M}_j$$ holds for all $j=1,2,\cdots,\gamma$ if and only if $$\mathbf{f}_1 - \mathbf{f}_2 \notin \bigcup\limits_{j=1}^{\gamma}\mathcal{L}_{1/2}(\mathbf{M}_j).$$
\end{lemma}

\begin{proof}
    It is sufficient to prove that there exists $j$ with $1 \leq j \leq \gamma$ such that $\mathbf{r}_{1,j} - \mathbf{r}_{2,j} \equiv \mathbf{r}_{2,j} - \mathbf{r}_{1,j} \mod \mathbf{M}_j$ holds if and only if $\mathbf{f}_1 - \mathbf{f}_2 \in \bigcup\limits_{j=1}^{\gamma}\mathcal{L}_{1/2}(\mathbf{M}_j)$.
    
    For any given $j$ with $1 \leq j \leq \gamma$, we can get 
    \begin{equation}\notag
        \mathbf{f}_1-\mathbf{f}_2 \equiv \mathbf{r}_{1,j} - \mathbf{r}_{2,j} \mod \mathbf{M}_j 
    \end{equation}
    and
    \begin{equation}\notag
         \mathbf{f}_2-\mathbf{f}_1 \equiv \mathbf{r}_{2,j} - \mathbf{r}_{1,j} \mod \mathbf{M}_j
    \end{equation}
    from (\ref{f1f2}).
    As s result, $$\mathbf{r}_{1,j} - \mathbf{r}_{2,j} \equiv \mathbf{r}_{2,j} - \mathbf{r}_{1,j} \mod \mathbf{M}_j$$ holds if and only if $$\mathbf{f}_{1} - \mathbf{f}_2 \equiv \mathbf{f}_2 - \mathbf{f}_1 \mod \mathbf{M}_j$$ holds, which is also equivalent to $$\mathbf{f}_1-\mathbf{f}_2 = \mathbf{M}_j \mathbf{n} + (\mathbf{f}_2-\mathbf{f}_1)$$ for some integer vector $\mathbf{n}$. Equivalently, $$\mathbf{f}_1 - \mathbf{f}_2 = \mathbf{M}_j(\frac{1}{2}\mathbf{n})$$ for some integer vector $\mathbf{n}$, i.e., $\mathbf{f}_1 - \mathbf{f}_2 \in \mathcal{L}_{1/2}(\mathbf{M}_j)$.

    Therefore, there exists $j$ with $1 \leq j \leq \gamma$ such that $\mathbf{r}_{1,j} - \mathbf{r}_{2,j} \equiv \mathbf{r}_{2,j} - \mathbf{r}_{1,j} \mod \mathbf{M}_j$ holds if and only if $\mathbf{f}_1 - \mathbf{f}_2 \in \bigcup\limits_{j=1}^{\gamma}\mathcal{L}_{1/2}(\mathbf{M}_j)$. 
\end{proof}

Next, we present our new condition to achieve the maximal possible dynamic range of MD-CRT for two integer vectors.

\begin{thm}\label{th:condition}
    Consider two integer vectors $\mathbf{f}_1$ and $\mathbf{f}_2$ and $\gamma$ non-singular integer matrices $\mathbf{M}_1,\mathbf{M}_2,\cdots,\mathbf{M}_{\gamma}$. Let $\mathbf{r}_{i,j}$ be the vector remainder of $\mathbf{f}_i$ modulo $\mathbf{M}_j$, for $i=1,2$ and $j=1,2,\cdots,\gamma$, and $\mathcal{S}_j=\{\mathbf{r}_{1,j},\mathbf{r}_{2,j}\}$ be the residue set of vector remainders of $\mathbf{f}_1$ and $\mathbf{f}_2$ modulo $\mathbf{M}_j$.  
    If 
    \begin{equation}\label{condition}
     \mathbf{f}_1 - \mathbf{f}_2 \ \text{or} \ \mathbf{f}_2 - \mathbf{f}_1  \in\Big{(} LAT(\mathbf{R}) + \bigcap\limits_{j=1}^{\gamma}\mathcal{N}(\mathbf{M}_j)\Big{)} \setminus \bigcup\limits_{j=1}^{\gamma}\mathcal{L}_{1/2}(\mathbf{M}_j), 
    \end{equation}
     where $\mathbf{R}$ is an lcrm of all the matrix moduli, then we can uniquely reconstruct $\mathbf{f}_1$ and $\mathbf{f}_2$ from all the vector residue sets $\mathcal{S}_j$, for $j=1,2,\cdots,\gamma$, when $\{\mathbf{f}_1,\mathbf{f}_2\}\subset \mathcal{N}(\mathbf{R}).$
\end{thm}

It is interesting to note that, since $LAT(\mathbf{R})$ is independent of a particular form of an lcrm $\mathbf{R}$ as shown in Prop. \ref{samelattice}, the condition (\ref{condition}) is also independent of the choice of an lcrm $\mathbf{R}$ of $\mathbf{M}_j$, $1\leq j\leq \gamma$. Before presenting the proof, we first give an example when $\gamma=2$ to show that if $\mathbf{k}$ belongs to the set in (\ref{condition}), $-\mathbf{k}$ may or may not be in the set
$$LAT(\mathbf{R})+\big{(}\mathcal{N}(\mathbf{M}_1)\cap \mathcal{N}(\mathbf{M}_2)\big{)},$$ if we make no assumptions about the matrices $\mathbf{M}_j$ for $j=1,2$. Example \ref{ex:notequal} below will help the following proof of Theorem \ref{th:condition}.

\begin{example}\label{ex:notequal}
    Let 
    \begin{equation}\notag
       \mathbf{M}_1 =  \begin{bmatrix}
            3 &0\\
            0 &4\\
        \end{bmatrix}, \quad 
        \mathbf{M}_2 =  \begin{bmatrix}
            3 &0\\
            0 &3\\
        \end{bmatrix}.
    \end{equation}
    Their lcrm is $\mathbf{R}=diag\{3,12\}$ and $\mathcal{N}(\mathbf{M}_1)\cap \mathcal{N}(\mathbf{M}_2) = \{0,1,2\}^{2}$. 
    Then, we have 
    \begin{multline*}
    [1, 0]^{\top},\ [0, 1]^{\top} \in 
   \Big( LAT(\mathbf{R}) + \big( \mathcal{N}(\mathbf{M}_1) 
    \cap \mathcal{N}(\mathbf{M}_2) \big) \Big) \\
   \setminus \big( \mathcal{L}_{1/2}(\mathbf{M}_1) 
    \cup \mathcal{L}_{1/2}(\mathbf{M}_2) \big).
    \end{multline*}
    For vector $[-1,0]^{\top}$, we have 
    \begin{equation}\notag
        \begin{bmatrix}
            -1 \\
            0 \\
        \end{bmatrix} =  \begin{bmatrix}
            3 &0\\
            0 &12\\
        \end{bmatrix} \begin{bmatrix}
            -1\\
            0\\
        \end{bmatrix}+ \begin{bmatrix}
            2\\
            0\\
        \end{bmatrix},
    \end{equation}
    which means $$[-1,0]^{\top} \in LAT(\mathbf{R})+\big{(}\mathcal{N}(\mathbf{M}_1)\cap \mathcal{N}(\mathbf{M}_2)\big{)}.$$
    For vector $[0,-1]^{\top}$, it has unique representation given $\mathbf{R}$:
    \begin{equation}\notag
        \begin{bmatrix}
            0 \\
           -1 \\
        \end{bmatrix} =  \begin{bmatrix}
            3 &0\\
            0 &12\\
        \end{bmatrix} \begin{bmatrix}
            0\\
            -1\\
        \end{bmatrix}+ \begin{bmatrix}
            0\\
            11\\
        \end{bmatrix},
    \end{equation}
    where $[0,11]^{\top} \in \mathcal{N}(\mathbf{R})$.
    Since $\mathcal{N}(\mathbf{M}_1)\cap \mathcal{N}(\mathbf{M}_2) \subset \mathcal{N}(\mathbf{R})$ and $[0,11]^{\top} \notin \mathcal{N}(\mathbf{M}_1)\cap \mathcal{N}(\mathbf{M}_2)$, we get that $$[0,-1]^{\top} \notin LAT(\mathbf{R})+\big{(}\mathcal{N}(\mathbf{M}_1)\cap \mathcal{N}(\mathbf{M}_2)\big{)}.$$
\end{example}

Now, we begin to prove Theorem \ref{th:condition}. The key of the following proof is to demonstrate that when (\ref{condition})
 is satisfied, we can determine the correspondence between the vector remainders in each residue set $\mathcal{S}_j$ and the vectors $\mathbf{f}_1,\mathbf{f}_2$ to determine . Then, we can use MD-CRT separately to uniquely determine them in $\mathcal{N}(\mathbf{R})$.
\begin{proof}
Due to symmetry, we only prove the case when $\mathbf{f}_1-\mathbf{f}_2$ is in the set in (\ref{condition}).
 Firstly, we show that 
\begin{align}\label{theq1}
&\Big( LAT(\mathbf{R}) + \bigcap\limits_{j=1}^{\gamma} \mathcal{N}(\mathbf{M}_j) \Big)
\setminus \bigcup\limits_{j=1}^{\gamma} \mathcal{L}_{1/2}(\mathbf{M}_j) \notag \\
&= \Big( LAT(\mathbf{R}) + \big( \bigcap\limits_{j=1}^{\gamma} 
\mathcal{N}(\mathbf{M}_j) \setminus \{ \mathbf{0} \} \big) \Big) 
\setminus \bigcup\limits_{j=1}^{\gamma} \mathcal{L}_{1/2}(\mathbf{M}_j),
\end{align}
where the condition on the right hand side makes it easier to apply Lemma \ref{lm:2} above directly.
Since 
$$LAT(\mathbf{R}) = \bigcap\limits_{j=1}^{\gamma}LAT(\mathbf{M}_j) \subset \bigcap\limits_{j=1}^{\gamma}\mathcal{L}_{1/2}(\mathbf{M}_j),$$
we can get 
$$LAT(\mathbf{R}) \subset \bigcup\limits_{j=1}^{\gamma}\mathcal{L}_{1/2}(\mathbf{M}_j),$$
which means (\ref{theq1}) holds.

Then, from Lemma \ref{lm:2}, there are exactly two vectors in each residue set $\mathcal{S}_j$. Since we don't know the correspondence between two vector remainders in each residue set $\mathcal{S}_j$ and the two vectors to determine, $\mathbf{f}_1,\mathbf{f}_2$, we use $\mathbf{v}_{j1}$ and $\mathbf{v}_{j2}$ to denote the two vectors in $\mathcal{S}_j$ for simplicity.    

For each residue set $\mathcal{S}_j$, we first calculate the vector remainder of $\mathbf{v}_{j1}-\mathbf{v}_{j2}$ modulo $\mathbf{M}_j$, denoted by $\mathbf{d}_{j1}$, and the vector remainder of $\mathbf{v}_{j2}-\mathbf{v}_{j1}$ modulo $\mathbf{M}_j$, denoted by $\mathbf{d}_{j2}$. From Lemma \ref{lm:3}, we can get that $\mathbf{d}_{j1} \neq \mathbf{d}_{j2}$, for all $j=1,2,\cdots,\gamma$, since $$\mathbf{f}_1 - \mathbf{f}_2 \notin \bigcup\limits_{j=1}^{\gamma}\mathcal{L}_{1/2}(\mathbf{M}_j).$$ It also means that we get two distinct vectors for each residue set $\mathcal{S}_j$ in the above procedure. Denote that $\mathcal{D}_j=\{\mathbf{d}_{j1},\mathbf{d}_{j2}\}$ for $j=1,2,\cdots,\gamma$.

Since $$\mathbf{f}_1 - \mathbf{f}_2 \in LAT(\mathbf{R}) + \bigcap\limits_{j=1}^{\gamma}\mathcal{N}(\mathbf{M}_j),$$
we can get that the $\gamma$ vector reminders of $\mathbf{r}_{1,j}-\mathbf{r}_{2,j}$ modulo $\mathbf{M}_j$ for all $1 \leq j \leq \gamma$ are the same from Lemma \ref{lm:1}. Note that the vector remainder of $\mathbf{r}_{1,j}-\mathbf{r}_{2,j}$ modulo $\mathbf{M}_j$ is one of $\mathbf{d}_{j1}$ and $\mathbf{d}_{j2}$. Since $\mathbf{d}_{j1} \neq \mathbf{d}_{j2}$ for all $j=1,2,\cdots,\gamma$, then we can find one vector $\mathbf{d}^{*}$ such that for each $j=1,2,\cdots,\gamma$, there exist exactly one vector $\mathbf{d}_{ji}$, $i=1\ \text{or}\ 2$, in $\mathcal{D}_j$ such that $\mathbf{d}_{ji}$= $\mathbf{d}^{*}$.

If \footnote{According to Example \ref{ex:notequal}, although $$\mathbf{f}_1 - \mathbf{f}_2 \in\Big{(} LAT(\mathbf{R}) + \bigcap\limits_{j=1}^{\gamma}\mathcal{N}(\mathbf{M}_j)\Big{)} \setminus \bigcup\limits_{j=1}^{\gamma}\mathcal{L}_{1/2}(\mathbf{M}_j),$$ 
we are not sure whether $$\mathbf{f}_2 - \mathbf{f}_1 \in LAT(\mathbf{R}) + \bigcap\limits_{j=1}^{\gamma}\mathcal{N}(\mathbf{M}_j)$$ 
or not.}
$$\mathbf{f}_2 - \mathbf{f}_1 \notin LAT(\mathbf{R}) + \bigcap\limits_{j=1}^{\gamma}\mathcal{N}(\mathbf{M}_j),$$
the vector remainders of $\mathbf{r}_{2,j}-\mathbf{r}_{1,j}$ modulo $\mathbf{M}_j$ for all $1 \leq j \leq \gamma$ are not the same according to Lemma \ref{lm:1}. As a result, the selected vector $\mathbf{d}^{*}$ is the vector remainder of $\mathbf{r}_{1,j}-\mathbf{r}_{2,j}$ modulo $\mathbf{M}_j$ for each j, $ 1\leq j\leq \gamma$. For each $j=1,2,\cdots,\gamma$, if $\mathbf{d}_{j1}=\mathbf{d}^{*}$, we can get that $\mathbf{v}_{j1}$ is the vector remainder $\mathbf{r}_{1,j}$ and $\mathbf{v}_{j2}$ is the vector remainder $\mathbf{r}_{2,j}$ since $\mathbf{d}_{j1}$ is the vector remainder of $\mathbf{v}_{j1}-\mathbf{v}_{j2}$ modulo $\mathbf{M}_j$. So, let $\mathbf{r}_{1,j}=\mathbf{v}_{j1}$ and $\mathbf{r}_{2,j}=\mathbf{v}_{j2}$. Similarly, if $\mathbf{d}_{j2}=\mathbf{d}^{*}$, we can get that $\mathbf{v}_{j2}$ is the vector remainder $\mathbf{r}_{1,j}$ and $\mathbf{v}_{j1}$ is the vector remainder $\mathbf{r}_{2,j}$ and let $\mathbf{r}_{1,j}=\mathbf{v}_{j2}$ and $\mathbf{r}_{2,j}=\mathbf{v}_{j1}$. 
In this way, we successfully distinguish, within each set $\mathcal{S}_j$, which vector corresponds to the vector remainder of $\mathbf{f}_1$ and which corresponds to the vector remainder of $\mathbf{f}_2$.

If $$\mathbf{f}_2 - \mathbf{f}_1 \in LAT(\mathbf{R}) + \bigcap\limits_{j=1}^{\gamma}\mathcal{N}(\mathbf{M}_j),$$
the vector remainders of $\mathbf{r}_{2,j}-\mathbf{r}_{1,j}$ modulo $\mathbf{M}_j$ for all $1 \leq j \leq \gamma$ are also the same according to Lemma \ref{lm:1}. It also means that all the sets $\mathcal{D}_j$ are the same. So, the selected $\mathbf{d}^{*}$ is either the vector remainder of $\mathbf{r}_{1,j}-\mathbf{r}_{2,j}$ modulo $\mathbf{M}_j$ or the vector remainder of $\mathbf{r}_{2,j}-\mathbf{r}_{1,j}$ modulo $\mathbf{M}_j$.
Since we actually care about the set $\{\mathbf{f}_1,\mathbf
{f}_2\}$ and don't care about which one is $\mathbf{f}_1$ and which one is the other,
we then do the same operations: if $\mathbf{d}_{j1}=\mathbf{d}^{*}$, let  $\mathbf{r}_{1,j}=\mathbf{v}_{j1}$ and $\mathbf{r}_{2,j}=\mathbf{v}_{j2}$; if $\mathbf{d}_{j2}=\mathbf{d}^{*}$, let  $\mathbf{r}_{1,j}=\mathbf{v}_{j2}$ and $\mathbf{r}_{2,j}=\mathbf{v}_{j1}$, for each $j=1,2,\cdots,\gamma$,. Then, we can get two sets $\{\mathbf{r}_{1,j} \mid j=1,2,\cdots,\gamma\}$ and $\{\mathbf{r}_{2,j} \mid j=1,2,\cdots,\gamma\}$, one of which is the set of all vector remainders of $\mathbf{f}_1$ and the other is the set of all vector remainders of $\mathbf{f}_2$. 

Combining the two cases above, we can always divide the vectors in all $\mathcal{S}_j$ into two sets, one of which is the set of all vector remainders of $\mathbf{f}_1$ and the other is the set of all vector remainders of $\mathbf{f}_2$. And in each set above, the correspondence between the vectors in the set and their respective matrix moduli is known.

Then, we can use use MD-CRT separately to solve the following two systems of congruence equations:
 \begin{equation}\label{twosystems}
  \begin{dcases}
    \begin{aligned}
      \mathbf{f}_1 &\equiv \mathbf{r}_{1,1} \mod{\mathbf{M}_1} \\
      \mathbf{f}_1 &\equiv \mathbf{r}_{1,2} \mod{\mathbf{M}_2} \\
       &\vdotswithin{\equiv}\\
      \mathbf{f}_1 &\equiv \mathbf{r}_{1,\gamma} \mod{\mathbf{M}_\gamma}
    \end{aligned}
  \end{dcases}
  \quad \text{and} \quad
  \begin{dcases}
    \begin{aligned}
      \mathbf{f}_2 &\equiv \mathbf{r}_{2,1} \mod{\mathbf{M}_1} \\
      \mathbf{f}_2 &\equiv \mathbf{r}_{2,2} \mod{\mathbf{M}_2} \\
      &\vdotswithin{\equiv}\\
      \mathbf{f}_2 &\equiv \mathbf{r}_{2,\gamma} \mod{\mathbf{M}_\gamma}
    \end{aligned}
  \end{dcases}
\end{equation}
and uniquely determine $\mathbf{f}_1$ and $\mathbf{f}_2$ when
  $\{\mathbf{f}_1,\mathbf{f}_2\}\subset \mathcal{N}(\mathbf{R})$.  
\end{proof}

A detailed algorithm for Theorem \ref{th:condition} is summarized in Algorithm \ref{alg:2vec} and we present an example to show how this algorithm works.

\begin{algorithm}
\caption{Reconstruction of Two Integer Vectors with Prior Information}
\label{alg:2vec}
\begin{algorithmic}[1]
\Require
$\gamma$ sets $\mathcal{S}_j=\{\mathbf{v}_{j1}, \mathbf{v}_{j2}\}$ and $\gamma$ non-singular integer matrices $\mathbf{M}_j$, $j=1,\ldots,\gamma$ with their lcrm $\mathbf{R}$
\Ensure
Two integer vectors $\mathbf{f}_1$ and $\mathbf{f}_2$

\For{each set $\mathcal{S}_j$}
    \State Compute vector remainders:
    \State $\mathbf{d}_{j1} \gets (\mathbf{v}_{j1} - \mathbf{v}_{j2}) \mod \mathbf{M}_j$
    \State $\mathbf{d}_{j2} \gets (\mathbf{v}_{j2} - \mathbf{v}_{j1}) \mod \mathbf{M}_j$
\EndFor

\State Find a common vector $\mathbf{d}^{*}$ such that $\forall j=1,\cdots,\gamma, \exists i\in\{1,2\}$ where $\mathbf{d}_{ji} = \mathbf{d}^{*}$

\For{each $j=1,2,\cdots,\gamma$}
    \If{$\mathbf{d}_{j1} = \mathbf{d}^{*}$}
        \State $\mathbf{r}_{1,j} \gets \mathbf{v}_{j1}$, $\mathbf{r}_{2,j} \gets \mathbf{v}_{j2}$
    \Else
        \State $\mathbf{r}_{1,j} \gets \mathbf{v}_{j2}$, $\mathbf{r}_{2,j} \gets \mathbf{v}_{j1}$
    \EndIf
\EndFor

\State Solve congruence systems (\ref{twosystems}) using MD-CRT
\State \Return $\mathbf{f}_1$, $\mathbf{f}_2$

\end{algorithmic}
\end{algorithm}

\begin{example}\label{ex:2}
   Consider the matrix moduli as follows:
    \begin{equation}\notag
\begin{aligned} 
\mathbf{M}_{1} &=
    \begin{pmatrix}
        4 & 1 \\
        1 & 1 \\
    \end{pmatrix},
    \mathbf{M}_{2} =
    \begin{pmatrix}
        3 & 3 \\
        1 & 2 \\
    \end{pmatrix},\\ 
    \mathbf{M}_{3} &=
    \begin{pmatrix}
        2 & 1 \\
        0 & 2 \\
    \end{pmatrix},
    \mathbf{M}_{4} =
    \begin{pmatrix}
        5 & 1 \\
        1 & 1 \\
    \end{pmatrix}.
\end{aligned}
\end{equation}

All the FPDs of these four matrices are shown in Fig. \ref{fig:ex}. One can easily see in the figure that vector 
$$[2,1]^{\top} \in \bigcap\limits_{j=1}^{4} \mathcal{N}(\mathbf{M}_j).$$
Besides, we can use the algorithm in \cite{MD2,remainder} to get an lcrm of these four matrices:
\begin{equation*}
  \mathbf{R}=
  \begin{pmatrix}
        12 & 0 \\
        0 & 12 \\
  \end{pmatrix}.
\end{equation*}
One can check that 
$$[2,1]^{\top} \in \Big{(} LAT(\mathbf{R}) + \bigcap\limits_{j=1}^{4}\mathcal{N}(\mathbf{M}_j)\Big{)} \setminus \bigcup\limits_{j=1}^{\gamma}\mathcal{L}_{1/2}(\mathbf{M}_j).$$
\begin{figure}[htbp]
    \centering
    \includegraphics[width=\columnwidth]{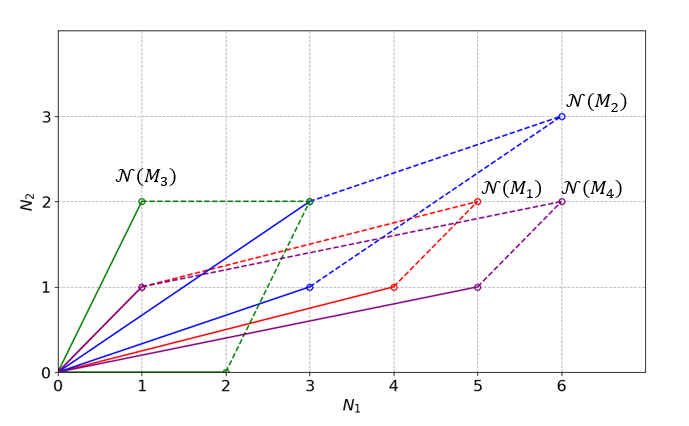}  
    \caption{FPDs of the four matrices in Example \ref{ex:2}}
    \label{fig:ex}  
\end{figure}

Let $\mathbf{f}_1=[10,7]^{\top}$ and $\mathbf{f}_2=[8,6]^{\top}$. They all in set $\mathcal{N}(\mathbf{R})$ and satisfy the condition (\ref{condition}) in Theorem \ref{th:condition}.
By calculating the vector remainders, we can get the following residue sets:
\begin{equation}\notag
\begin{aligned}
    \mathcal{S}_1 &= \{[0,0]^{\top},[3,1]^{\top}\}, \
    \mathcal{S}_2 = \{[4,2]^{\top},[2,1]^{\top}\},\\
    \mathcal{S}_3 &= \{[1,1]^{\top},[1,0]^{\top}\}, \
    \mathcal{S}_4 = \{[4,1]^{\top},[3,1]^{\top}\},\\
\end{aligned}    
\end{equation}
where the first vector in $\mathbf{S}_j$ is the vector remainder of $\mathbf{f}_1$ modulo $\mathbf{M}_j$ and the second vector is the vector remainder of $\mathbf{f}_2$ modulo $\mathbf{M}_j$, for $j=1,2,3,4$.

Next, we show how to use Algorithm \ref{alg:2vec} to determine $\mathbf{f}_1$ and $\mathbf{f}_2$.
Without loss of generality, the first vector in $\mathcal{S}_j$ is denoted by $\mathbf{v}_{j1}$ and the second is denoted by $\mathbf{v}_{j2}$. Then, we can get the following sets $\mathcal{D}_j$ by calculating $\mathbf{d}_{j1}$ and $\mathbf{d}_{j2}$:
\begin{equation}\notag
\begin{aligned}
    \mathcal{D}_1 &= \{[2,1]^{\top},[3,1]^{\top}\},\
    \mathcal{D}_2 = \{[2,1]^{\top},[4,2]^{\top}\},\\
    \mathcal{D}_3 &= \{[2,1]^{\top},[1,1]^{\top}\}, \
    \mathcal{D}_4 = \{[2,1]^{\top},[4,1]^{\top}\},\\
\end{aligned}    
\end{equation}
where the first vector in $\mathcal{D}_j$ is $\mathbf{d}_{j1}$ and the second vector is $\mathbf{d}_{j2}$. Then, we can find $\mathbf{d}^{*}=[2,1]^{\top}$ such that there is exactly one vector in each set $\mathcal{D}_j$ that is equal to $\mathbf{d}^{*}$ for $j=1,2,3,4$. Since $\mathbf{d}_{j1}=\mathbf{d}^{*}$ for all $j=1,2,3,4$, we can get that $\mathbf{r}_{1,1}=[0,0]^{\top},\mathbf{r}_{1,2}=[4,2]^{\top},\mathbf{r}_{1,3}=[1,1]^{\top},\mathbf{r}_{1,4}=[4,1]^{\top}$ and $\mathbf{r}_{2,1}=[3,1]^{\top},\mathbf{r}_{2,2}=[2,1]^{\top},\mathbf{r}_{2,3}=[1,0]^{\top},\mathbf{r}_{2,4}=[3,1]^{\top}$. Then, by solving   
\begin{equation*}
      \mathbf{f}_i \equiv \mathbf{r}_{i,j} \mod{\mathbf{M}_j}, \ \text{for} \ j=1,2,3,4,
\end{equation*}
for $i=1,2$ and
we can get $\mathbf{f}_1=[10,7]^{\top}$ and $\mathbf{f}_2=[8,6]^{\top}$.
\end{example}
 
\subsection{Comparisons with Previous Results in One-Dimensional Case}

For one dimensional generalized CRT, conditions on unknown integers to determine from their residue sets modulo several moduli have been obtained \cite{zhou,xiao15} to achieve the maximal possible dynamic range, i.e., the lcm of all the moduli. Interestingly, in this subsection, we show that the new condition we have obtained in the previous subsection for generalized MD-CRT is weaker than all the known ones in \cite{zhou,xiao15},  when the dimension D reduces to 1.

Let $N_1 < N_2$ be two distinct positive integers to be determined, and $1<m_1<m_2<\cdots<m_{\gamma}$ be $\gamma$ integers as moduli. Let $r_{i,j}$ be the remainder of $N_i$ modulo $m_j$ for $i=1,2$ and $j=1,2,\cdots,\gamma$, and $\mathcal{B}_j=\{r_{1,j},r_{2,j}\}$ be the set of remainders of $N_1$ and $N_2$ modulo $m_j$, for $j=1,2,\cdots,\gamma$. In the classical CRT, when we use $m_1,m_2,\cdots,m_{\gamma}$ as moduli, we can uniquely reconstruct a positive integer, if it is lower than the $lcm(m_1,m_2,\cdots,m_{\gamma})$, denoted by $M$, which is also the maximal possible dynamic range for generalized CRT of two integers. In general, this maximal possible dynamic range may not be achieved for two integers as we have analyzed before. However, under some conditions on the two unknown integers, it may be achieved as studied in \cite{zhou,xiao15} with detailed determination algorithms. These conditions are stated below.

\begin{prop}\label{prop:cond1}\cite{zhou} 
    If $N_2-N_1 < m_1/2$, then we can uniquely determine $N_1$ and $N_2$ in $[0,M)$ from the sets $\mathcal{B}_j$ for $j=1,2,\cdots,\gamma$.
\end{prop}

\begin{prop}\label{prop:cond2}\cite{xiao15}
    If $N_2-N_1 < m_1$, $gcd(2,m_j)=1$ for $j=1,2,\cdots,\gamma$, $\gamma \geq 2$, then we can uniquely determine $N_1$ and $N_2$ in $[0,M)$ from the sets $\mathcal{B}_j$ for $j=1,2,\cdots,\gamma$.
\end{prop}

\begin{prop}\label{prop:cond3}\cite{xiao15}
    If $N_2-N_1 < m_1$, $gcd(2,m_j)=1$ for $j=1,2,\cdots,\gamma-1$, and $m_{\gamma}-2m_1>0$, then we can uniquely determine $N_1$ and $N_2$ in $[0,M)$ from the sets $\mathcal{B}_j$ for $j=1,2,\cdots,\gamma$.
\end{prop}

Our new condition, as presented below, can be derived from Theorem \ref{th:condition}. Accordingly, the algorithm for determining the two integers can be directly obtained from Algorithm \ref{alg:2vec}.

\begin{cor}
    If $N_2-N_1 \in \{1,2,\cdots,m_1-1,M-m_1+1,\cdots,M-2,M-1\}\setminus \{km_1/2,km_2/2,\cdots,km_{\gamma}/2\mid k \in \mathbb{Z}\}$, then we can uniquely determine $N_1$ and $N_2$ in $[0,M)$ from the sets $\mathcal{B}_j$ for $j=1,2,\cdots,\gamma$.
\end{cor}

Next, we show that our newly obtained condition is much weaker (better) than the above existing conditions. Firstly, when we use the same group of moduli, our prior condition on $N_2-N_1$ is weaker than that of all three propositions above.

The condition in Prop. \ref{prop:cond1} is equivalent to $N_2-N_1 \in \{1,2,\cdots,\lceil m_1/2\rceil -1\}$. Since $$\lceil m_1/2\rceil -1<m_1/2<m_2/2<\cdots<m_{\gamma}/2,$$
we can get
\begin{align*}
\{1, 2, \cdots, \lceil m_1& / 2 \rceil - 1\} 
\subset 
\{1,\cdots, m_1 - 1, M - m_1 + 1,\\ &\cdots, M - 1\} 
\setminus \{km_1/2, \cdots, km_{\gamma}/2 \mid k \in \mathbb{Z}\},
\end{align*}
where the right hand side is that in Theorem \ref{th:condition} when the vector dimension is 1.

As for the above Props. \ref{prop:cond2} and \ref{prop:cond3}, the prior condition on $N_2-N_1$ is equivalent to $N_2-N_1 \in \{1,2,\cdots,m_1-1\}$ with an additional condition on moduli that is not needed in our new result in Theorem \ref{th:condition}. When the condition on moduli in Prop. \ref{prop:cond2} or Prop. \ref{prop:cond3} is satisfied, we have
$$ \{1,2,\cdots,m_1-1\} \cap \{km_1/2,\cdots,km_{\gamma}/2 \mid k \in \mathbb{Z}\} = \varnothing.$$
As a result, when the same group of moduli that satisfy the condition in Prop. \ref{prop:cond2} or Prop. \ref{prop:cond3} is chosen, we have 
\begin{align*}
\{1, 2, \cdots, m_1- 1\} 
&\subset 
\{1,\cdots, m_1 - 1, M - m_1 + 1, \cdots,\\& M - 1\} 
\setminus \{km_1/2, \cdots, km_{\gamma}/2 \mid k \in \mathbb{Z}\}.
\end{align*}

\begin{example}
   Consider a group of moduli $\{5,7,9,11\}$ that satisfies both conditions in Props. \ref{prop:cond2} and \ref{prop:cond3}. When this group is used, the prior condition on $N_2-N_1$ in Prop. \ref{prop:cond1} becomes $$N_2-N_1\in\{1,2\};$$ the prior condition on $N_2-N_1$ in Props. \ref{prop:cond2} and \ref{prop:cond3} becomes $$N_2-N_1\in\{1,2,3,4\};$$ the prior condition on $N_2-N_1$ in our result becomes $$N_2-N_1\in\{1,2,3,4,3461,3462,3463,3464\}.$$
\end{example}

In some applications of CRT, such as modulo samplers used in self-reset analog-to-digital converters (SR-ADC) \cite{lugan2}, we usually choose moduli in a given size, i.e., the maximal moduli is less than a given positive integer. The dynamic range (the lcm of all the moduli) depends on the choices of moduli and we usually want to achieve a dynamic range as large as possible. Next, we show that our result provides much more possible choices of moduli, thus making it possible to achieve a larger dynamic range while ensures that the prior condition on $N_2-N_1$ provided by our new result is weaker than those in \cite{zhou,xiao15}. Since there is no constraint on moduli in Prop. \ref{prop:cond1}, we only consider Props. \ref{prop:cond2} and \ref{prop:cond3} below by an example.

\begin{example}
    Assume that all the moduli are less than 12 and only $4$ moduli can be chosen.

    In Prop. \ref{prop:cond2}, if we want to achieve the largest possible dynamic range, the moduli we choose are $5,7,9,11$. In this case, the largest possible dynamic range is $3465$ and the condition on $N_2-N_1$ is $$N_2-N_1\in\{1,2,3,4\}.$$

    In Prop. \ref{prop:cond3}, if we want to achieve the largest possible dynamic range, the moduli we choose are also $5,7,9,11$. So, the dynamic range and the condition on $N_2-N_1$ are the same with that of Prop. \ref{prop:cond2}.

    In our new condition, if we want to achieve the largest possible dynamic range, the moduli we choose are $7,9,10,11$, and the largest dynamic range is $6930$. The condition $N_2-N_1$ becomes $$N_2-N_1\in\{1,2,3,4,6,6924,6926,6927,6928,6929\}.$$

    In comparison, our new result can achieve a larger dynamic range as well as a larger range for $N_2-N_1$.
\end{example}

The improvement over the previous results is achieved because, from the perspective of lattices, we establish a broader sufficient condition under which the size of each residue set precisely matches the number of vectors to be determined.

\section{Conclusion}\label{s5}

In this paper, we have presented a generalized MD-CRT for multiple integer vectors and investigated two key questions. First, we have established a uniquely determinable range for multiple integer vectors and proposed an efficient algorithm for their determinations. Then, we have presented a novel condition for achieving the maximal possible dynamic range when dealing with two integer vectors. Notably, our new condition is strictly weaker than all the previously known ones when reduced to the one-dimensional case. However, determining the condition for achieving the maximal possible dynamic range in the case of more than two unknown vectors remains open.


\end{document}